%% file: main.tex
\documentclass[runningheads]{llncs}

\include{preamble_FoSSaCS}

\include{macros_FoSSaCS}

\begin{document}
	\title{Proof Systems for the Modal $\mu$-Calculus Obtained by Determinizing Automata} 
	\titlerunning{Proof Systems for the Modal $\mu$-Calculus}
	%
	\author{Maurice Dekker \and
		Johannes Kloibhofer \and
		Johannes Marti \and
		Yde Venema}
	\authorrunning{M. Dekker et al.}
	%
	\institute{ILLC, University of Amsterdam, Netherlands}
	\maketitle              
	\input{sec.Abstract}

\input{sec.Introduction}

\input{sec.Preliminaries}

\input{sec.Determinisation}

\input{sec.ProofSystem}

	\input{sec.Conclusion}

	%
	%
 	\bibliographystyle{splncs04}
	\bibliography{bib.Tableaux}
	
	\appendix
	\include{sec.Appendix}

\end{document}

%% file: preamble_FoSSaCS.tex
\usepackage[T1]{fontenc}
\usepackage{graphicx}
\usepackage[hidelinks]{hyperref}
\usepackage{color}

\usepackage{textalpha} 
\usepackage{amsmath,amssymb, latexsym}
\usepackage{ebproof} 
\usepackage{xspace} 
\usepackage{xargs} 
\usepackage{mathdots} 
\usepackage{subcaption} 
\usepackage[all]{xypic}
\usepackage{refcount} 

\usepackage{tikz} 
\usetikzlibrary{automata, positioning} 

%% file: macros_FoSSaCS.tex
\newcommand{\calD}{\mathcal{D}}

\newcommand{\calL}{\mathcal{L}}

\newcommand{\calO}{\mathcal{O}}
\newcommand{\calP}{\mathcal{P}}

\newcommand{\calT}{\mathcal{T}}

\newcommand{\bbA}{\mathbb{A}} 
\newcommand{\bbB}{\mathbb{B}}

\newcommand{\sfA}{\mathsf{A}} 
 
\newcommand{\sfP}{\mathsf{P}} 

\newcommand{\sfR}{\mathsf{R}}
\newcommand{\sfS}{\mathsf{S}}
\newcommand{\sfT}{\mathsf{T}} 
\newcommand{\sff}{\mathsf{f}} 

\renewcommand{\|}{~|~}

\newcommand{\powset}{\mathcal{P}}

\newcommand{\ran}{\mathrm{ran}}

\newcommand{\sop}{[} 
\newcommand{\scl}{]} 
\newcommand{\sel}[2]{#1 \backslash #2} 
\newcommand{\unsubst}[2]{\sop \sel{#1}{#2} \scl} 
\newcommand{\Liff}{\Leftrightarrow} 
\newcommand{\Impl}{\Rightarrow} 

\renewcommand{\land}{\wedge}
\renewcommand{\lor}{\vee}
\newcommand{\proves}{\vdash}

\newcommand{\AxLit}{\ensuremath{\mathsf{Ax1}}\xspace}
\newcommand{\AxTop}{\ensuremath{\mathsf{Ax2}}\xspace}
\newcommand{\RuOr}{\ensuremath{\mathsf{R}_{\lor}}\xspace}
\newcommand{\RuAnd}{\ensuremath{\mathsf{R}_{\land}}\xspace}
\newcommand{\RuBox}{\ensuremath{\mathsf{\mathsf{R}_{\Box}}}\xspace}

\newcommand{\RuFp}[1]{\ensuremath{\mathsf{R}_{#1}}\xspace}
\newcommand{\RuMu}{\RuFp{\mu}}

\newcommand{\RuNu}{\RuFp{\nu}}

\newcommand{\RuDischarge}[1][\dx]{\ensuremath{\mathsf{D}^{#1}}\xspace}

\newcommand{\RuResolve}{\ensuremath{\mathsf{Resolve}}\xspace}
\newcommandx{\RuCompress}[2][1= ,2= ]{\ensuremath{\mathsf{Compress}_{#1}^{#2}}\xspace}

\newcommand{\Tokens}{\mathcal{D}}

\newcommand{\dx}{\ensuremath{\mathsf{x}}}
\newcommand{\dy}{\ensuremath{\mathsf{y}}}

\newcommand{\upto}{\upharpoonright}

\newcommand{\BT}{\ensuremath{\mathsf{BT}}\xspace}
\newcommand{\BTInf}{\ensuremath{\mathsf{BT}^{\infty}}\xspace}

\newcommand{\JS}{\ensuremath{\mathsf{JS}}\xspace}
\newcommand{\NW}{\ensuremath{\mathsf{NW}}\xspace}

\newcommand{\Focus}{\ensuremath{\mathsf{Focus}}\xspace}

\newcommand{\leaves}{\mathrm{leaves}}
\newcommand{\tree}{\mathrm{tree}}
\newcommand{\green}{\ensuremath{\mathrm{green}}\xspace}
\newcommand{\white}{\ensuremath{\mathrm{white}}\xspace}
\newcommand{\red}{\ensuremath{\mathrm{red}}\xspace}
\newcommand{\TSeq}{\mathrm{TSeq}}
\newcommand{\minL}{\mathrm{minL}}

\newcommand{\dia}{\Diamond}

\newcommand{\atneg}[1]{\overline{#1}}

\newcommand{\isbnf}{\;::=\;}
\newcommand{\divbnf}{\;\mid\;}

\newcommand{\Prop}{\mathsf{Prop}}

\newcommand{\muML}{\mathcal{L}_{\mu}}

\newcommand{\Clos}{\mathsf{Clos}}
\newcommand{\Fix}{\mathsf{Fix}}

\renewcommand{\phi}{\varphi}

\newcommand{\tracestep}{\to_{C}}

%% file: sec.Abstract.tex
\begin{abstract}
Automata operating on infinite objects feature prominently in the theory of the modal $\mu$-calculus. 
One such application concerns the tableau games introduced by Niwi\'{n}ski \& Walukiewicz, 
of which the winning condition for infinite plays can be naturally checked by a nondeterministic parity stream automaton.
Inspired by work of Jungteerapanich and Stirling we show how determinization constructions of this automaton may be used to directly obtain proof systems for the $\mu$-calculus.
More concretely, we introduce a  binary tree construction for determinizing nondeterministic parity stream automata. 
Using this construction we define the annotated cyclic proof system \BT, where formulas are annotated by tuples of binary strings. 
Soundness and Completeness of this system follow almost immediately from the correctness of the determinization method.	
	
	\keywords{modal mu-calculus  
		\and derivation system 
		\and determinisation of B\"uchi and parity automata
		\and non-wellfounded and cyclic proofs
	}
\end{abstract}

%% file: sec.Introduction.tex
\section{Introduction}\label{sec.intro}

\paragraph{The modal $\mu$-calculus} The modal $\mu$-calculus is a
natural extension of basic modal logic with explicit least and greatest
fixpoint operators.
Allowing the formulation of various recursive phenomena,
this extension raises the expressive power of the language (at least when it
comes to 
bisimulation-invariant properties of transition systems) to that of monadic
second-order logic \cite{Janin96}. The $\mu$-calculus is generally regarded 
as a universal specification language, since it embeds most other logics that
are used for this purpose, such as LTL, CTL, CTL$^*$ and PDL. 
Despite its expressive power the $\mu$-calculus has still reasonable
computational properties; its model checking problem is in
quasi-polynomial time \cite{Calude17} and its satisfiability problem is
\textsc{exptime}-complete \cite{Emerson99}. Another interesting
feature of the theory of the modal $\mu$-calculus lies in its
connections with other fields, in particular the theory of
finite automata operating on infinite objects, and that of infinite
games.

\paragraph{Derivation systems}
Given the importance of the modal $\mu$-calculus, there is a natural interest 
in the development and study of derivation systems for its validities.
And indeed, already in \cite{Kozen1983} Kozen proposed an axiomatization.
Despite the naturality of this axiom system, he only established a partial 
completeness result, and it took a substantial amount of time before 
Walukiewicz \cite{Walukiewicz2000} managed to prove soundness and 
completeness for the full language.

Kozen's axiomatization amounts to a Hilbert-style derivation system, making 
it less attractive for proof search.
If one is interested in a cut-free system, a good starting point is  the 
two-player tableau-style game introduced by Niwi\'{n}ski \& Walukiewicz 
\cite{Niwinski1996}.
Here we will present their system in the shape of a derivation system $\NW$
(this change of perspective can be justified by identifying winning 
strategies for one of the players in the game with $\NW$-proofs).
$\NW$ is a one-sided sequent system which allows for infinite proofs: although 
its proof rules are completely  standard (and finitary), due to the unfolding 
rules for the fixpoint operators, derivations may have infinite branches.
A crucial aspect of the $\NW$-system is that one has to keep track of the 
\emph{traces} of individual formulas along the infinite branches.
A derivation will only count as a proper proof if each of its infinite branches
is \emph{successful}, in the sense that it carries a so-called $\nu$-trace: 
a trace which is dominated by a \emph{greatest} fixpoint operator.

This condition is easy to formulate but not so nice to work with.
One could describe the subsequent developments in the proof theory for the modal 
$\mu$-calculus as a series of modifications of the system $\NW$ which aim to 
get a grip on the complexities and intricacies of the above-mentioned traces,
and in particular, to use the resulting ``trace management'' for the introduction of 
finitary, cyclic proof systems.
Landmark results were obtained by Jungteerapanich \cite{Jungteerapanich2010}
and Stirling \cite{Stirling2014}, who introduced cyclic proof systems for the 
$\mu$-calculus, two calculi that we will identify here under the name $\JS$.

\paragraph{Automata \& Derivation systems}

Applications of automata theory are ubiquitous in the theory of the modal 
$\mu$-calculus, and the area of proof theory is no exception.
In particular, Niwi\'{n}ski \& Walukiewicz \cite{Niwinski1996} observed that 
infinite matches of their game, corresponding to infinite branches in an 
$\NW$-derivation, can be seen as infinite words or \emph{streams} over some 
finite alphabet.
It follows that \emph{stream automata} (automata operating on infinite words) 
can be used to determine whether such a match/branch carries a $\nu$-trace.
Niwi\'nski \& Walukiewicz used this perspective to link 
their results to the exponential-time complexity of the satisfiability problem 
for the $\mu$-calculus. 

A key contribution of Jungteerapanich and Stirling
\cite{Jungteerapanich2010,Stirling2014} was to bring automata \emph{inside} the 
proof system.
The basic idea would be to decorate each sequent in a derivation with a state
of the stream automaton which recognizes whether an infinite branch is successful
or not; starting from the root, the successive states decorating the sequents 
on a given branch simply correspond to a run of the automaton on this branch.
For this idea to work one needs the stream automaton to be
\emph{deterministic}.
To see this, observe that two successful but distinct branches in a derivation 
would generally require two distinct runs, and in the case of a nondeterministic 
automaton, these two runs might already diverge before the two branches split.
 
Interestingly, there is a natural stream automaton recognizing the successful
branches of an $\NW$-derivation:
One may simply take the states of such an automaton to be the formulas
in the (Fischer-Ladner) \emph{closure} of the root sequent.
But given the \emph{nondeterministic} format of this automaton, before it can be 
used in a proof system, we need to transform it into an equivalent deterministic 
one. 
This explains the relevance of constructions for determinizing stream automata  
to the proof theory of the modal $\mu$-calculus.

\paragraph{Determinization of stream automata}
Using the ideas we just sketched, one may obtain sound and complete derivation
systems for the modal $\mu$-calculus in an easy way.
For any deterministic automaton $\bbA$ that recognizes the successful branches 
in $\NW$-derivations, one could simply introduce new-style sequents 
consisting of an $\NW$-sequent decorated with a state of $\bbA$, and 
adapt the proof rules of $\NW$ incorporating the transition map of $\bbA$.
This could be done in such a way that the stream of decorations of an infinite 
branch corresponds to the run of $\bbA$ on the 
stream of sequents of the same branch.
The trace condition of $\NW$-derivations could then be replaced by the 
acceptance condition of $\bbA$ (which is generally much simpler, since it does
not refer to traces).

More interesting is to use specific determinization constructions, in 
order to design more attractive proof systems or to prove results \emph{about} 
the derivation system (and thus, potentially, about the $\mu$-calculus).
In particular, some determinization constructions are based on a \emph{power 
construction}, meaning that the states of the deterministic automaton consist
of \emph{macrostates} (\emph{subsets} of the nondeterministic original) with
some additional structure.
Such constructions allow for proof calculi where this additional structure is 
incorporated into the sequents.
For instance, the derivation system $\JS$ is based on the well-known Safra 
construction  \cite{Safra1988}, in which the states of the deterministic 
automaton consist of macrostates of the original automaton that are organised
by means of so-called \emph{Safra trees}.
Concretely, the (augmented) sequents in $\JS$ consist of a set of 
\emph{annotated formulas}, with the annotations indicating the position of the
formula in the Safra tree and a so-called \emph{control} which provides 
additional information on the Safra tree.

\paragraph{Our contribution}
Our overall goal is to explicitize the role of automata theory in the design 
of derivation systems for the modal $\mu$-calculus (and other fixpoint
logics).
Our point is that distinct determinization constructions lead to distinct 
sequent system, and that we may look for alternatives to the Safra construction.
Concretely the contribution of this paper is threefold:
\begin{enumerate}
\item
We provide a new determinization construction for both B\"uchi and parity stream
automata which is based on binary trees. 
Our construction is similar to constructions related to so-called profile trees
\cite{Fogarty2015,Loeding2019}.
\item
We apply our construction to obtain a new derivation system $\BT$
for the modal $\mu$-calculus.
While our system is similar in spirit to the system $\JS$, a key difference is 
that our sequents consist of annotated formulas, and nothing else.
\item
We establish the soundness and completeness of $\BT$.
A distinguishing feature of our approach is that (up to some optimizations)
this result is a \emph{direct} consequence of the soundness and completeness 
of $\NW$ and the adequacy of our determinization construction.
\end{enumerate}

\paragraph{Related work}
There is an extensive literature on applications of automata theory in the theory of the modal $\mu$-calculus, among others \cite{Doumane2017,Janin1995,Janin96,Wilke2001}. Jungteerapanich and Stirling \cite{Jungteerapanich2010,Stirling2014} were the first to obtain an annotated proof system inspired by the determinization of automata.
The proof system \Focus for the alternation-free $\mu$-calculus designed by Marti \& Venema \cite{Marti2021} originates with a rather simple determinization 
construction for so-called weak automata.
In \cite{Leigh2023}, Leigh \& Wehr also take a rather general approach towards
the use of determinization constructions in the design of derivation systems, 
but they confine attention to the Safra construction.

 \paragraph{Overview of paper}
 In the next section we provide the necessary background material on binary 
 trees, on $\omega$-automata, on the modal $\mu$-calculus and the proof system
\NW ; doing so we fix our notation.
 In Section~\ref{sec.det} we introduce a new determinization method for 
 nondeterministic Büchi and parity automata.
 We will use this construction to prove the soundness and completeness of the 
 proof system \BT, which we introduce in Section \ref{sec.BTproofs}. 
 All missing proofs can be found in the appendix.


%% file: sec.Preliminaries.tex
\section{Preliminaries}\label{sec.prelim}

\paragraph{Binary trees}

We let $2^*$ denote the set of \emph{binary strings}; we write $<$ for the
lexicographical order of $2^{*}$, and $\sqsubseteq$ for the (initial) substring
relation given by $s \sqsubseteq t$ if $sr = t$ for some $r$.
\emph{Substitution} for binary strings is defined in the following way: 
Let $s,t,\tilde{s},r \in 2^*$ be such that $s = t \tilde{s}$, then  $s\unsubst{t}{r}$ denotes the binary string $r \tilde{s}$.
A \emph{binary tree} is a finite set of binary strings $T \subseteq 2^*$ such 
that $s0 \in T \Impl s \in T$ and $s0 \in T \Liff s1 \in T$.
Here we let $\leaves(T) = \{s \in T ~|~s0 \notin T\}$ denote its set of 
\emph{leaves}, and $\minL(T)$ its \emph{minimal} leaf of $T$, i.e. the unique 
leaf of the form $0\cdots 0$. 
A set of binary strings $L$ is a \emph{set of leaves of a binary trees} if for 
all $s\neq t \in L$ we have $s \not\sqsubseteq t$ and $\tree(L) = \{s \in
2^* ~|~ \exists t \in L: s \sqsubseteq t \}$ is a binary tree.

\paragraph{Stream automata}

A \emph{non-deterministic automaton} over a finite alphabet $\Sigma$ is a 
quadruple $\bbA = \langle A, \Delta, a_I,\mathrm{Acc}\rangle$, where $A$ is a 
finite set, $\Delta: A \times \Sigma \rightarrow  \powset(A)$ is the transition 
function of $\bbA$, $a_I \in A$ its initial state and $\mathrm{Acc} 
\subseteq A^{\omega}$ its acceptance condition. 	
An automaton is called \emph{deterministic} if $|\Delta(a,y)| = 1$ for all pairs 
$(a,y) \in A \times \Sigma$.
A \emph{run} of an automaton $\bbA$ on a stream $w=y_0y_1y_2... \in 
\Sigma^{\omega}$ is a stream $a_0a_1a_2... \in A^{\omega}$ such that $a_0 = a_I$ 
and $a_{i+1} \in \Delta(a_i,y_i)$ for all $i \in \omega$. 
A stream $w$ is \emph{accepted} by $\bbA$ if there is a run of $\bbA$ on $w$, 
which is in $\mathrm{Acc}$; we define $\calL(\bbA)$ to be the set of all
accepting streams of $\bbA$.

The acceptance condition can be given in different ways:
A \emph{Büchi} condition is given as a subset $F \subseteq A$. 
The corresponding acceptance condition is the set of runs, which contain 
infinitely many states in $F$.
A \emph{parity} condition is given as a map $\Omega: A \rightarrow \omega$. The corresponding acceptance condition is the set of runs $\alpha$ such that $\min\{\Omega(a) \| a \text{ occurs infinitely often in } \alpha \}$ is even.
A \emph{Rabin} condition is given as a set $R = ((G_i,B_i))_{i \in I}$ of pairs of subsets of $A$.  The corresponding acceptance condition is the set of runs $\alpha$ for which there exists $i \in I$ such that $\alpha$ contains infinitely many states in $G_i$ and finitely many in $B_i$.
Automata with these acceptance conditions are called \emph{Büchi}, \emph{parity} 
and \emph{Rabin automata}, respectively.

\paragraph{Modal $\mu$-calculus: Syntax}
The set $\muML$ of \emph{formulas} of the modal $\mu$-calculus is generated 
by the grammar
\[
\phi \isbnf 
p \divbnf \atneg p 
\divbnf \bot \divbnf \top
\divbnf (\phi\lor\phi) \divbnf (\phi\land\phi) \divbnf
\dia\phi \divbnf \Box\phi 
\divbnf \mu x . \phi \divbnf \nu x . \phi,
\]
where $p$ and $x$ are taken from a fixed set $\Prop$ of propositional
variables and in formulas of the form $\mu x. \phi$ and $\nu x.
\phi$ there are no occurrences of $\atneg x$ in $\phi$. 

Formulas of the form $\mu x . \phi$ ($\nu x . \phi$) are called 
\emph{$\mu$-formulas} (\emph{$\nu$-formulas}, respectively); formulas 
of either kind are called \emph{fixpoint formulas}.
We write $\eta, \lambda \in \{\mu,\nu\}$ to denote an arbitrary fixpoint operator.
We use standard terminology and notation for the binding of variables by 
the fixpoint operators and for substitutions, and make sure only to 
apply substitution in situations where no variable capture will occur. 
An important use of the substitution operation concerns the \emph{unfolding}
$\chi[\xi/x]$ of a fixpoint formula $\xi = \eta x . \chi$.

Given two formulas $\phi,\psi \in \muML$ we write $\phi \tracestep \psi$ if 
$\psi$ is either a direct boolean or modal subformula of $\phi$, or else 
$\phi$ is a fixpoint formula and $\psi$ is its unfolding.
The \emph{closure} $\Clos(\Phi) \subseteq \muML$ of $\Phi \subseteq \muML$
is the least superset of $\Phi$ that is closed under this relation.
It is well known that $\Clos(\Phi)$ is finite iff $\Phi$ is finite.
A \emph{trace} is a sequence $(\phi_{n})_{n<\kappa}$, with $\kappa \leq
\omega$, such that $\phi_{n} \tracestep \phi_{n+1}$, for all $n
+ 1 < \kappa$.

We define a \emph{dependence order} on the fixpoint formulas occurring in $\Phi$, written $\Fix(\Phi)$, by setting $\eta x. \phi 
<_{\Phi} \lambda y. \psi$ (where smaller in $<_{\Phi}$ means being of higher 
priority) if $\Clos(\eta x.\phi) = \Clos(\lambda y.\psi)$ and $\eta x.\phi$ is a subformula of $\lambda y.\psi$.
One may define a \emph{parity function} $\Omega: \Fix(\Phi) \rightarrow
\omega$, which respects this order (i.e., $\Omega(\eta x. \phi)<\Omega(\lambda y.
\psi)$ if $\eta x. \phi <_{\Phi} \lambda y \psi$) and satisfies
$\Omega(\eta x. \phi)$ is even iff $\eta = \nu$. 
Let $\max_\Omega(\Phi) = \max\{\Omega(\nu x. \phi)\|\nu x. \phi \in 
\Fix(\Phi)\}$.

It is well known that any infinite trace $\tau = (\phi_{n})_{n<\kappa}$ features a unique formula 
$\phi$ that occurs infinitely often on $\tau$ and is
a subformula of $\phi_{n}$ for cofinitely many $n$. This formula is always a 
fixpoint formula, and where it is of the form $\eta x.\psi$ we call $\tau$ an
\emph{$\eta$-trace}.

Since the semantics of the modal $\mu$-calculus only plays an indirect role in 
our paper, we refrain from giving the definition.

\paragraph{Non-wellfounded proofs}

A sequent $\Gamma$ is a finite set of $\mu$-calculus formulas, possibly with 
additional structure such as annotations.
Rules have the following form, possibly with additional side conditions:
\bigskip

\begin{minipage}{\textwidth}
	\begin{minipage}{0.40\textwidth}
		\begin{prooftree}
			\hypo{\Gamma_1}
		    \hypo{\cdots}
			\hypo{\Gamma_n}
			\infer[left label=$R$: ]3[]{\Gamma}
		\end{prooftree}
	\end{minipage}
	\begin{minipage}{0.40\textwidth}
		\begin{prooftree}
			\hypo{[\Gamma]^{\dx}}
			\infer[no rule]1{{\vdots}}
			\infer[no rule]1{\Gamma}
			\infer[left label= \RuDischarge[\dx]: ~~]1[]{\Gamma}
		\end{prooftree}
\end{minipage}
\end{minipage}

\bigskip
\noindent
A rule $R$, where $n=0$, is called an axiom. 
The rules \RuDischarge[\dx] are called \emph{discharge} rules.  
Each discharge rule is marked by a unique \emph{discharge token} taken from a fixed infinite set $\calD = \{\dx,\dy,...\}$.

\begin{definition} 
A \emph{derivation system} $\calP$ is a set of rules. 
A $\calP$ \emph{derivation} $\pi = (T,P,\sfS,\sfR, \sff)$ is a quintuple such that
		$(T,P)$ is a, possibly infinite, tree with nodes $T$ and parent relation $P$;
$\sfS$ is a function that maps every node $u \in T$ to a non-empty sequent 
$\Sigma_u$;
$\sfR$ is a function that maps every node $u \in T$ to its \emph{label} 
$\sfR(u)$, which is either (i) the name of a rule in $\calP$ or 
(ii) a discharge token; and $\sff$ is a partial function that maps some nodes $u \in T$ to its \emph{principal formula} $\sff(u) \in \sfS(u)$. If a specific formula $\phi$ in the conclusion of a rule is designated, then $\sff(u) = \phi$ and otherwise $\sff(u)$ is undefined. 
To qualify as a derivation, such a quintuple is required to satisfy the following conditions:
\begin{enumerate}
\item
If a node is labeled with the name of a rule then it has as many children 
as the rule has premises, and the annotated sequents at the node and its 
children match the specification of the rules.
						
\item
If a node is labeled with a discharge token then it is a leaf.
For every leaf $l$ that is labeled with a discharge token $\dx \in \Tokens$ 
there is exactly one node $u \in T$ that is labeled with \RuDischarge[\dx].
This node $u$ and its child are proper ancestors of $l$. 
In this situation we call $l$ a \emph{discharged leaf}, and $u$ its
\emph{companion}; we write $c$ for the function that maps a discharged
leaf $l$ to its companion $c(l)$ and write $p(l)$ for the path in $T$ from $c(l)$ to $l$. 
\end{enumerate}

\end{definition}
A derivation $\pi' = (T',P',\sfS',\sfR', \sff')$ is a \emph{subderivation} of $\pi = (T,P,\sfS,\sfR, \sff)$ if $(T',P')$ is a subtree of $(T,P)$ and $\sfS',\sfR', \sff'$ and $\sfS,\sfR, \sff$ are equal on $(T',P')$.
A derivation $\pi$ is called \emph{regular} if it has finitely many distinct subderivations.
	
\begin{definition}
Let $\pi = (T,P,\sfS,\sfR, \sff)$ be a derivation. We define two graphs we are interested 
in: (i) The usual \emph{proof tree} $\calT_\pi = (T,P)$ and (ii) the \emph{proof tree
with back edges} $\calT_\pi^C = (T,P^C)$, where $P^C = P \cup \{(l,c(l))\mid l 
\text{ is a discharged leaf}\}$ is the parent relation plus back-edges for every 
discharged leaf.
	
A \emph{strongly connected subgraph} in $\calT_\pi^C$ is a set $S$ of nodes, such
that for every $u,v \in S$ there is a $P^{C}$-path from $u$ to $v$.
\end{definition}

\paragraph{The \NW proof system}
The rules of the derivation system \NW, which is directly based on the tableau games introduced by Niwi\'{n}ski \& Walukiewicz \cite{Niwinski1996}, are given in Figure~\ref{fig.NWrules}.

\begin{figure}[thb]
	\begin{minipage}{\textwidth}
		\begin{minipage}{0.20\textwidth}
			\begin{prooftree}
				\hypo{\phantom{X}}
				\infer[left label= \AxLit]1{p, \bar{p}, \Gamma}
			\end{prooftree}
		\end{minipage}
		\begin{minipage}{0.20\textwidth}
			\begin{prooftree}
				\hypo{\phantom{X}}
				\infer[left label= \AxTop]1{\top, \Gamma}
			\end{prooftree}
		\end{minipage}
		\begin{minipage}{0.24\textwidth}
			\begin{prooftree}
				\hypo{\varphi,\psi,\Gamma}
				\infer[left label = \RuOr]1{\varphi \lor \psi,\Gamma}
			\end{prooftree}
		\end{minipage}
		\begin{minipage}{0.30\textwidth}
			\begin{prooftree}
				\hypo{\varphi, \Gamma}
				\hypo{\psi,\Gamma}
				\infer[left label = \RuAnd]2{\varphi \land \psi,\Gamma}
			\end{prooftree}
		\end{minipage}
	\end{minipage}
	
	\bigskip
	
	\begin{minipage}{\textwidth}
		\begin{minipage}{0.30\textwidth}
			\begin{prooftree}
				\hypo{ \varphi,\Gamma}
				\infer[left label= \RuBox]1[]{ \Box \varphi, 	\Diamond \Gamma, \Delta}
			\end{prooftree}
		\end{minipage}
		\begin{minipage}{0.30\textwidth}
			\begin{prooftree}
				\hypo{\varphi[\mu x . \varphi / x], \Gamma}
				\infer[left label = \RuMu]1{\mu x . \varphi, \Gamma}
			\end{prooftree}
		\end{minipage}
		\begin{minipage}{0.22\textwidth}
			\begin{prooftree}
				\hypo{\varphi[\nu x . \varphi / x], \Gamma}
				\infer[left label = \RuNu]1{\nu x . \varphi, \Gamma}
			\end{prooftree}
		\end{minipage}
	\end{minipage}
	
	\caption{Rules of \NW}
	\label{fig.NWrules}
\end{figure}

In order to decide whether an \NW derivation qualifies as a proper \emph{proof}, 
one has to keep track of the development of individual formulas along infinite 
branches of the proofs.

\begin{definition}
Let $\Gamma,\Gamma'$ be sequents, $\xi$ be a formula such that $\Gamma$ is the
conclusion and $\Gamma'$ is a premise of a rule in Figure \ref{fig.NWrules} 
with principal formula $\xi$. 
We define the \emph{active} and \emph{passive trail relation}
$\sfA_{\Gamma,\xi,\Gamma'}, \sfP_{\Gamma,\xi,\Gamma'} \subseteq 
\Gamma \times \Gamma'$. 
Both relations are defined via a case distinction on $\xi$:
	
	\emph{Case $\xi = \Box\phi$:} Then $\Gamma = \Box\phi, \Diamond \Lambda, \Delta$ and $\Gamma' =
	\varphi, \Lambda$. We define $\sfA_{\Gamma,\xi,\Gamma'}, 
	= \{(\Box \phi,\phi)\} \cup \{(\Diamond \chi, \chi) \mid \chi \in
	\Lambda\}$ and $\sfP_{\Gamma,\xi,\Gamma'} = \varnothing$.
	
	\emph{Case $\xi = \varphi \lor \psi$:} Then $\Gamma = \varphi \lor \psi, \Lambda$ and $\Gamma' = \varphi,\psi, \Lambda$. We define $\sfA_{\Gamma,\xi,\Gamma'} = \{(\varphi \lor
	\psi,\varphi),(\varphi \lor \psi,\psi)\}$ and $\sfP_{\Gamma,\xi,\Gamma'} = \{(\chi,\chi) \mid \chi \in
	\Lambda\}$.

	
The relations for the remaining rules are defined analogously.
	
The \emph{trail relation} $\sfT_{\Gamma,\xi,\Gamma'}\subseteq \Gamma \times 
\Gamma'$ is defined as $\sfA_{\Gamma,\xi,\Gamma'} \cup \sfP_{\Gamma,\xi,\Gamma'}$.
Finally, for nodes $u,v$ in an \NW proof $\pi$, such that $P(u,v)$, we define 
$\sfT_{u,v} = \sfT_{\sfS(u),\sff(u),\sfS(v)}$
\end{definition}

Note that for any two nodes $u,v$ with $P(u,v)$ and $(\phi,\psi) \in \sfT_{u,v}$,
we have either $(\phi,\psi) \in \sfA_{{u,v}}$ and $\phi \tracestep \psi$, 
or else $(\phi,\psi) \in \sfP_{{u,v}}$ and $\phi = \psi$.
The idea is that $\sfA$ connects the active formulas in the premise and 
conclusion, whereas $\sfP$ connects the side formulas. 

\begin{definition}
Let $\pi = (T,P,\sfS,\sfR, \sff)$ be an \NW derivation. 
A \emph{branch} of $\pi$ is simply a (finite or infinite) branch of the underlying 
tree $(T,P)$ of $\pi$.
	%
A \emph{trail} on a branch $\alpha = (v_{n})_{n<\kappa}$ is a sequence 
$\tau = (\phi_{n})_{n<\kappa}$ of formulas such that $(\phi_{i},\phi_{i+1}) \in 
\sfT_{v_i,v_{i+1}}$, whenever $i+1 < \kappa$.
We obtain the \emph{tightening} $\widehat{\tau}$ of such a $\tau$ by erasing
all $\phi_{i+1}$ from $\tau$ for which $(\phi_{i},\phi_{i+1})$ belongs to the
passive trail relation $\sfP_{v_{i},v_{i+1}}$.
We call $\tau$ a \emph{$\nu$-trail} if its tightening $\widehat{\tau}$ is 
a $\nu$-trace (and so, in particular, it is infinite).
\end{definition}



\begin{definition}
An \emph{\NW proof} $\pi$ is an \NW derivation such that on every infinite branch of $\pi$ there is a $\nu$-trail. We write $\NW \vdash \Gamma$ if there is an \NW proof of $\Gamma$, i.e., an \NW proof, where $\Gamma$ is the sequent at the root of the proof. 
\end{definition}

Soundness and Completeness of \NW for guarded formulas, (i.e., where in every
subformula $\eta x . \psi$ all free occurrences of $x$ in $\psi$ are in the 
scope of a modality) follows from the results by Niwi\'{n}ski \&
Walukiewicz \cite{Niwinski1996}. 
As pointed out in \cite{Afshari2017}, it follows from \cite{Studer2008} and
\cite{Friedmann2013} that the result in fact holds for arbitrary formulas.
By Theorem~6.3 in \cite{Niwinski1996}, \NW-proofs can be assumed to be regular,
and this observation applies to unguarded formulas as well.

\begin{theorem}[Soundness \& Completeness]\label{thm.NWSoundCompleteness}
Let $\Gamma$ be a sequent, then $\bigvee \Gamma$ is valid iff 
$\NW \proves \Gamma$ iff
$\Gamma$ has a regular \NW-proof.
\end{theorem}

%
%

%% file: sec.Determinisation.tex
\section{Determinization of automata with binary trees}\label{sec.det}

\subsection{Büchi automata}
Let $\Sigma$ be an alphabet and $\bbB =  \langle B, {\Delta}, b_I, F \rangle$ 
a nondeterministic Büchi automaton over $\Sigma$. 
We want to present an equivalent deterministic Rabin automaton.  

The \emph{run tree} of $\bbB$ on a word $w = (w_i)_{i \in \omega}$ is a pair $\sfR = (R,l)$, where $R$
is the full infinite binary tree and $l$ labels every node $s$ with $B_s \subseteq B$,
such that $l(\epsilon) = \{b_I\}$ and for $|s| = i$: $l(s1) = \Delta(B_s,w_i)
\cap F$ and $l(s0) = \Delta(B_s,w_i) \cap \overline{F}$, where we define $\Delta(B_s,y) = \bigcup_{b\in B_s} \Delta(b,y)$. 
It describes all possible runs of $\bbB$ on $w$, using the 1s to keep track of
visited states in $F$.

The \emph{profile tree}, introduced in \cite{Fogarty2013}, is a pruned version
of the run tree, where 1) at each level all but the (lexicographically) greatest 
occurrence of a state $b$ are removed and 2) nodes labelled by the empty set are 
deleted. 
This results in a tree of bounded width, where every node has $0$,$1$ or $2$ 
children.
It can be proved that $\bbB$ accepts a stream $w$ iff the corresponding profile
tree has a branch with infinitely many 1s.

In \cite{Fogarty2015} a determinization method was defined, where macrostates
encode levels of the profile tree. 
In our approach macrostates encode a compressed version of the whole profile
tree up to some level: Nodes $u$, $v$ are identified (iteratively), if $v$ is
the unique child of $u$. This results in finite binary trees, where leaves
are labelled by subsets of $B$. In every step of the transition function
we add one level of the run tree and then prune and compress the tree to
obtain a binary tree again. 
Whenever a $1$ is compressed (in the sense of a node being identified with its right child) we know that a state in $F$ has been visited and mark the node green. A run of the deterministic automaton is accepted if there is a node, which never gets removed and is marked green infinitely often. 
Figure~\ref{fig:determinization} contains an example of this determinization 
construction.

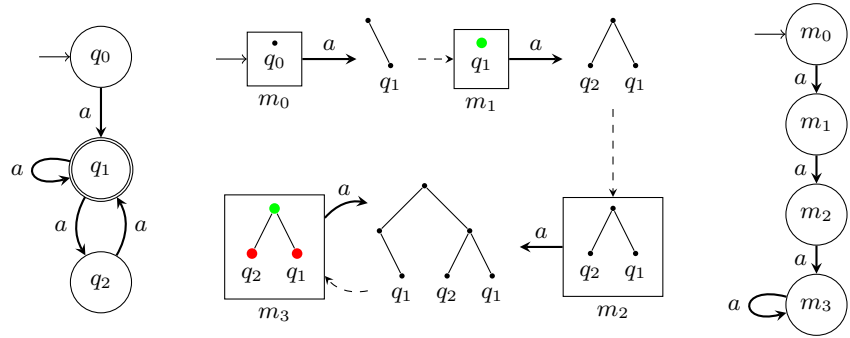
\begin{figure}[t]
\begin{subfigure}[c]{0.23\linewidth}
\begin{center}
	\begin{tikzpicture} [node distance = 1.5cm, 
		on grid, 
		auto,
		every loop/.style={stealth-}]
		
		\node (q0) [state, 
		initial, 
		initial text = {}] {$q_0$};
		
		\node (q1) [state,
		accepting,
		below = of q0] {$q_1$};
		
		\node (q2) [state,
		below = of q1] {$q_2$};
		
		\path [-stealth, thick]
		(q0) edge[left] node {$a$}   (q1)
		(q1) edge[bend right,left] node {$a$}   (q2)
		(q2) edge[bend right,right] node {$a$}   (q1)
		(q1) edge [loop left]  node {$a$}();
	\end{tikzpicture}	
\end{center}
\end{subfigure}
\begin{subfigure}[c]{0.52\linewidth}
\begin{center}
\begin{tikzpicture}

\node [draw, rectangle, initial, initial text={}, label=below:$m_0$] (init) at (0,0) {
\begin{tikzpicture}
  \coordinate
  {node [circle,fill=black,inner sep=0pt,minimum
size=2pt,label=below:$q_0$] {}} ;
\end{tikzpicture}
};

\node [] (initToFirst) at (1.5,0) {
\begin{tikzpicture}
  [level distance=6mm,
   level 1/.style={sibling distance=6mm}]
  \node [circle,fill=black,inner sep=0pt,minimum size=2pt] {}
     child [missing]
     child {node [circle,fill=black,inner sep=0pt,minimum size=2pt,label=below:$q_1$] {}};
\end{tikzpicture}
};

\node [draw, rectangle, label=below:$m_1$] (first) at (2.75,0) {
\begin{tikzpicture}
  \coordinate
  {node [circle,fill=green,inner sep=0pt,minimum
size=4pt,label=below:$q_1$] {}} ;
\end{tikzpicture}
};

\node [] (firstToSecond) at (4.5,0) {
\begin{tikzpicture}
  [level distance=6mm,
   level 1/.style={sibling distance=6mm}]
  \node [circle,fill=black,inner sep=0pt,minimum size=2pt] {}
     child {node [circle,fill=black,inner sep=0pt,minimum size=2pt,label=below:$q_2$] {}}
     child {node [circle,fill=black,inner sep=0pt,minimum size=2pt,label=below:$q_1$] {}};
\end{tikzpicture}
};

\node [draw, rectangle, label=below:$m_2$] (second) at (4.5,-2.5) {
\begin{tikzpicture}
  [level distance=6mm,
   level 1/.style={sibling distance=6mm}]
  \node [circle,fill=black,inner sep=0pt,minimum size=2pt] {}
     child {node [circle,fill=black,inner sep=0pt,minimum size=2pt,label=below:$q_2$] {}}
     child {node [circle,fill=black,inner sep=0pt,minimum size=2pt,label=below:$q_1$] {}};
\end{tikzpicture}
};

\node [] (preThird) at (2.25,-2.5) {
\begin{tikzpicture}
  [level distance=6mm,
   level 1/.style={sibling distance=12mm},
   level 2/.style={sibling distance=6mm}]
  \node [circle,fill=black,inner sep=0pt,minimum size=2pt] {}
     child {node [circle,fill=black,inner sep=0pt,minimum size=2pt] {}
       child [missing]
       child {node [circle,fill=black,inner sep=0pt,minimum
size=2pt,label=below:$q_1$] {}}
     }
     child {node [circle,fill=black,inner sep=0pt,minimum size=2pt] {}
       child {node [circle,fill=black,inner sep=0pt,minimum size=2pt,label=below:$q_2$] {}}
       child {node [circle,fill=black,inner sep=0pt,minimum size=2pt,label=below:$q_1$] {}}
     };
\end{tikzpicture}
};

\node [draw, rectangle, label=below:$m_3$] (third) at (0,-2.5) {
\begin{tikzpicture}
  [level distance=6mm,
   level 1/.style={sibling distance=6mm}]
  \node [circle,fill=green,inner sep=0pt,minimum size=4pt] {}
     child {node [circle,fill=red,inner sep=0pt,minimum size=4pt,label=below:$q_2$] {}}
     child {node [circle,fill=red,inner sep=0pt,minimum size=4pt,label=below:$q_1$] {}};
\end{tikzpicture}
};
		
		\path [-stealth, thick]
		(init) edge [above] node {$a$} (initToFirst)
		(first) edge [above] node {$a$} (firstToSecond)
		(second) edge [above] node {$a$} (preThird)
		(third) edge [above, bend left]  node {$a$} (preThird);

\path [-stealth, dashed]
 (initToFirst) edge (first)
 (firstToSecond) edge (second)
 (preThird) edge[bend left] (third);
\end{tikzpicture}	
\end{center}
\end{subfigure}
\begin{subfigure}[c]{0.23\linewidth}
\begin{center}
	\begin{tikzpicture} [node distance = 1.2cm, 
		on grid, 
		auto,
		every loop/.style={stealth-}]
		
		\node (m0) [state, 
		initial, 
		initial text = {}] {$m_0$};
		
		\node (m1) [state,
		below = of m0] {$m_1$};
		
		\node (m2) [state,
		below = of m1] {$m_2$};

		\node (m3) [state,
		below = of m2] {$m_3$};
		
		\path [-stealth, thick]
		(m0) edge[left] node {$a$}   (m1)
		(m1) edge[left] node {$a$}   (m2)
		(m2) edge[left] node {$a$}   (m3)
		(m3) edge[loop left] node {$a$}(m3);
	\end{tikzpicture}	
\end{center}
\end{subfigure}
\caption{A nondeterministic B\"{u}chi automaton $\bbB$ on the left and
its determinization $\bbB^D$ on the right. The diagram in the middle
shows the internal structure of the macrostates $m_0$, $m_1$, $m_2$ and
$m_3$ of $\bbB^D$. Binary trees are represented in the obvious way
(i.e., the root is the string $\epsilon$, and for every node the left
child appends $0$ and the right child appends $1$). The transitions of
$\bbB^D$ are split in two parts: In the first part one level of the run
tree is added, corresponding to the steps 1 and 2 in the definition of
the transition function. In the second part (the dashed arrows) the tree
is pruned and compressed, corresponding to the steps 3 and 4. The
acceptance condition of $\bbB^D$ is such that the word $a^\omega$ is
accepted by $\bbB^D$ because the string $\epsilon$ is always in play,
marked green infinitely often and never red.}
\label{fig:determinization}
\end{figure}

Formally we define the deterministic Rabin automaton $\bbB^D = \langle B^D, 
\delta, b'_I, R \rangle$ as follows: 
An element $S$ in the carrier $B^D$ of $\bbB^D$ is called a \emph{macrostate} 
and consists of 
\begin{itemize}
	\item a subset $B_S$ of $B$,
	\item a map $f: B_S \rightarrow 2^*$, such that\footnote{Here $\ran(f)$ denotes the co-domain of $f$.} $\ran(f)$ is a set of leaves of a binary tree and 
	\item a colouring map $c: \tree(\ran(f)) \rightarrow \{\green, \red, \white\}$.
\end{itemize}
We define $T^S$ to be the binary tree $\tree(\ran(f))$, that has $\ran(f)$ as its leaves and say that 
a binary string $s$ is \emph{in play} if $s \in T^S$. 
If it is clear from the context we occasionally abbreviate $T^S$ by $T$.
We will sometimes denote a macrostate by a set of pairs $(b,s)$, usually
written as $b^s$, where $b \in B_S$ and $s = f(b)$ and deal with the colouring $c$ implicitly. 

The initial macrostate $b'_I$ consists of the singleton $\{b_I^{\epsilon}\}$, 
where $c(\epsilon) = \white$. 
To define the transition function $\delta$ let $S$ be in $B^D$ and $y \in \Sigma$. 
We define $\delta(S,y)= S'$, where starting from the empty set we build up $S'$
in the following steps:
\begin{enumerate}
\item \underline{Move:} For every $a^{s} \in S$ and $b \in \Delta(a,y)$, add $b^{s}$ to $S'$.
\item \underline{Append:} For every $a^s \in S'$, where $a \notin F$, change $a^s$ to $a^{s0}$. For every $a^s \in S'$, where $a \in F$, change $a^s$ to $a^{s1}$.
\item \underline{Resolve:} 
If $a^{s}$ and $a^{t}$ are in $S'$, where $s<t$, delete $a^{s}$. 
\item \underline{Compress/Colour:} 
Let $c(t) = \white$ for every $t \in T^{S'}$. 
Now we compress and colour $T$ in the following way, until there exists no
\emph{witness} $t \in T$, such that (a) or (b) is 
applicable:\footnote{%
     As shown in Proposition \ref{prop.deltaProp} in the appendix this procedure does not depend on the order in which 
	 witnesses are chosen, and thus produces a unique binary tree.}
	 
\begin{enumerate}
\item For any $t \in T$, such that $t0 \in T$ and $t1 \notin T$, change every $a^s \in S'$, where $t0 \sqsubseteq s$, to $a^{s\unsubst{t0}{t}}$. For any $s \in T$, where $t \sqsubset s$, let $c(s) = \red$. 
\item For any $t \in T$, such that $t0 \notin T$ and $t1 \in T$, change every $a^s \in S'$, where $t1 \sqsubseteq s$, to $a^{s\unsubst{t1}{t}}$. For any $s \in T$ such that $t = s0\cdots 0$, let $c(s) = \green$, if $c(s) \neq \red$. In particular let $c(t) = \green$ if $c(t) \neq \red$. For any $s \in T$, where $t \sqsubset s$, let $c(s) = \red$. 
\end{enumerate}
	
\end{enumerate}
We define $B^{D}$ as the set of macrostates that can be reached from $b'_{I}$
in this way.

A run of $\bbB^D$ is accepting if there is a binary string $s$, which is in play
cofinitely often such that $c(s)$ is $\green$ infinitely often and $\red$ only 
finitely often.


\begin{theorem}\label{thm.CorrectnessNBtoDR}
	$\bbB$ accepts a word $w$ iff $\bbB^D$ accepts $w$.
\end{theorem}

\begin{remark}
For a Büchi automaton of $n$ states, our construction yields a deterministic
automaton $\bbB^D$ with $n^{\calO(n)}$ states and a Rabin condition of 
$\calO(2^n)$ pairs, see Lemma \ref{lem.complexityNBtoDR} in the appendix.
With some adaptations we could also match the optimal Rabin condition, which
is known to be linear-size~\cite{Safra1988}.

 This can be achieved by adding an labelling function as follows: Let $L = \{1,...,2n-1\}$ be the set of potential labels. Macrostates are defined as before, where an additional injective function $l: T^S \rightarrow L$ is added. For the initial state we let $l(\epsilon)=1$. The steps 1 - 4 in the transition function remain the same, where we add a final step 5 in which we define the new labeling function $l'$: We let $l'(s) = l(s)$ for all $s$ that already occurred in $T^{S}$ and for all $s \in T^{S'}\setminus T^S$ we let $c(s) = \red$ and choose new, distinct labels in $L$, i.e. ones which do not occur in $\ran(l)$. The binary tree $T^{S'}$ has at most $n$ leaves, hence it has at most $2n -1$ many nodes and this is always possible. 
	
	The new acceptance condition has the following form: A run of the automaton is accepting if there is a label $k \in L$, such that $c(l^{-1}(k))$ is $\green$ infinitely often and $\red$ only finitely often. Here $c(l^{-1}(k))$ is defined to be  $\red$ if $k \notin \ran(l)$. This is a Rabin condition with $\calO(n)$ pairs. Notably we still have $n^{\calO(n)}$ macrostates, thus the determination method is optimal.
\end{remark}

\subsection{Parity automata}\label{sec.sub.parityAut}

We now extend the approach to parity automata. 
Let $\Sigma$ be an alphabet and
$\bbA = \langle A, \Delta_A, a_I, \Omega \rangle$ be a nondeterministic parity 
automaton. 

In order to present the intuitive idea behind the construction we first 
transform $\bbA$ into an equivalent nondeterministic Büchi automaton $\bbB$. 
Let $m$ be the maximal even priority of $\Omega$. For even $k = 0,2,...m$ we define $\bbA_0,\bbA_2,...,\bbA_n$ as copies of $\bbA$ without the states of priority smaller than $k$, i.e. $\bbA_k = \langle A_k, \Delta_k, F_k \rangle$ with $A_k = \{a_k ~|~ a \in A \land  \Omega(a) \geq k\}$, $\Delta_k = \Delta_A|_{A_k}$ and $F_k = \{a_k \in A_k ~|~ \Omega(a) = k\}$. 
Now we define the nondeterministic Büchi automaton $\bbB =  \langle B, 
{\Delta_B}, b_I, F \rangle$:\footnote{%
   For easier notation we represent the transition function 
   $B\times\Sigma \rightarrow \powset(B)$ by its corresponding relation (i.e., 
   subset of $B\times\Sigma\times B$).}
\begin{align*}
	B = & A \cup \bigcup_{\substack{k=0\\k~ \text{even}}}^m A_k, \quad \quad\quad\quad \quad b_I = a_I, \quad \quad\quad \quad\quad F = \bigcup_{\substack{k=0\\k~ \text{even}}}^m F_k,\\ 
	{\Delta_B} = & \Delta_A \cup \bigcup_{\substack{k=0\\k~ \text{even}}}^m \Delta_k \cup \{(a,y,b_k) \in A \times \Sigma \times A_k ~|~ b \in \Delta_A(a,y), k = 0,2,...,m\}. 
\end{align*}
Although $\bbA_k$ is not an automaton, as it does not have an initial state, we can define the Büchi automaton $\bbA \cup \bbA_k = \langle A\cup A_k, \Delta_B|_{ A\cup A_k}, a_I, F_k \rangle$ for $k = 0,...,m$.

The intuition behind the determinization of the parity automaton $\bbA$
is the following: We apply the binary tree construction to every
automaton $\bbA \cup \bbA_k$ for $k=0,2,...,m$, which is possible as
there are no paths from $A_k$ to $A_j$ if $k \neq j$ and none of the
accepting states of $\bbB$ are in the set $A$.
The annotation of a state $a \in \bbA$ will then be the tuple $(s_0,s_2,...,s_m)$, where $s_k$ is the annotation at the state $a_k \in \bbA \cup \bbA_k$. 
Note that the automaton $\bbA^D$ will be different from the automaton obtained 
from the binary tree construction on the whole $\bbB$.

\medskip
To make that formal we need some definitions. 
A \emph{treetop} $L$ is a set of leaves of a binary tree, where potentially the 
minimal leaf is missing, i.e. $L$ is a finite set of binary strings such that 
for all $s\neq t \in L$ it holds $s \not\sqsubseteq t$ and $\tree(L) = \{s \in 2^* ~|~ \exists t \in L: s \sqsubseteq t \} \cup \{s0 \| s = 0\cdots0 \text{ and } s1 \in L\}$ is a binary tree.

For even $m$ let $\TSeq(m) = \{(s_0,s_{2},...,s_m) ~|~  s_0,s_{2},...,s_m \in 2^*\}$ be the set of sequences of length $\frac{m}{2}+1$, where $s_0,...,s_m$ are binary strings.
Let $\pi_k$ be the projection function, which maps $\sigma = (s_0,...,s_m)$ to $s_k$ for $k = 0,2,...,m$. 
We define a partial order $<$ on $\TSeq(m)$: Let $(s_0,...,s_m) < (t_0,...,t_m)$ if there exists $l \in \{0,...,m\}$ such that $s_l < t_l$ and $s_j = t_j$ for $j= 0,...,l-2$.

We now define the deterministic Rabin automaton $\bbA^D = \langle A^D, \delta_A, a'_I, R_A \rangle$. Let $m$ be the maximal even priority of $\Omega$ in $\bbA$.
An element $S$ in the carrier $A^D$ of $\bbA^D$ consists of a tuple $(A_S,f, c_0,...,c_m)$, where
\begin{itemize}
	\item  $A_S$ is a subset of $A$,
	\item $f: A_S \rightarrow \TSeq(m)$, such that $\ran(\pi_k \circ f)$ is a treetop for $k = 0,...,m$ and
	\item  $c_k$ is a colouring map from  $\tree(\ran(\pi_k \circ f)) \rightarrow \{\green, \red, \white\}$ for $k=0,2,...,m$.
\end{itemize}
We define $T^S_k$ to be the binary tree $\tree(\ran(\pi_k \circ f))$ for $k=0,2,...,m$ and say a binary string $s$ is \emph{in play at position} $k$ if $s \in T^S_k$. If the context is clear we will abbreviate $T^S_k$ with $T_k$.
Again we sometimes denote a macrostate by a set of pairs
$(a,\sigma)$, usually written as $a^{\sigma}$, where $a \in A_S$ and $\sigma = f(a)$ and deal with the colourings $c_k$ implicitly.

The initial macrostate $a'_I$ consists of the singleton $\{a_I^{(\epsilon,...,\epsilon)}\}$. To define the transition function $\delta_A$ let $S$ be in $A^D$ and $y \in \Sigma$. We define $\delta_A(S,y)= S'$, where $S'$ is constructed in the following steps:
\begin{enumerate}
\item 
\begin{enumerate}
\item \underline{Move:} For every $a^{\sigma} \in S$ and $b \in \Delta_A(a,y)$, add $b^{\sigma}$ to $S'$. 
\item \underline{Reduce:} For every $a^{\sigma} \in S'$, change $a^{\sigma}$ to $a^{\sigma'}$, where $\sigma'$ is obtained from $\sigma = (s_0,...,s_m)$ by replacing every $s_j$ with $j > \Omega(a)$ by $\minL(T_j)$.
\end{enumerate}
\item \underline{Append:} 
For every $a^{\sigma} \in S'$ and $\sigma = (s_0,...,s_m)$, change $a^{\sigma}$
to $a^{\sigma'}$, where 
$\sigma' = (s_00,...,s_{k-2}0,s_k1, s_{k+2}0,...,s_m0)$ if $\Omega(a) = k$ is 
even, and $\sigma' = (s_00,...,s_{m}0)$ if $\Omega(a) = k$ is odd.


\item \underline{Resolve:} If $a^{\sigma}$ and $a^{\tau}$ are in $S'$ and $\sigma < \tau$, delete $a^{\sigma}$. 
\item \underline{Compress/Colour:} Do for every $k = 0,2,...,m$: Let $c_k(t) = \white$ for any $t \in T_k$. Now we compress and colour $T_k$ inductively in the following way, until there exists no \emph{witness} $t \in T_k$, such that (a) or (b) is applicable:
\begin{enumerate}
		\item For any $t \in T_k$, such that $t0 \in T_k$ and $t1 \notin T_k$, change every $a^{\sigma} \in S'$, where $\sigma = (s_0,...,s_m)$, and $t0 \sqsubseteq s_k$, to $a^{\sigma'}$, where $\sigma' = (s_0,...,s_k\unsubst{t0}{t},...,s_m)$. For any $s \in T_k$, where $t \sqsubset s$, let $c_k(s) = \red$. 
		\item For any $t \in T_k$, such that $t0 \notin T_k$, $t1 \in T_k$ and $t \neq 0\cdots0$, change every $a^{\sigma} \in S'$, where $\sigma = (s_0,...s_m)$, and $t1 \sqsubseteq s_k$, to $a^{\sigma'}$, where $\sigma' = (s_0,...,s_k\unsubst{t1}{t},...,s_m)$. For any $s \in T_k$ such that $t = s0\cdots 0$, let $c_k(s) = \green$, if $c_k(s) \neq \red$. In particular let $c_k(t) = \green$ if $c_k(t) \neq \red$. For any $s \in T_k$, where $t \sqsubset s$, let $c_k(s) = \red$. 
\end{enumerate}

\end{enumerate}
A run of $\bbA^D$ is accepting if there is $k \in \{0,2,...,m\}$ and a binary
string $s$, which is in play at position $k$ cofinitely often such that $c_k(s)$
is $\green$ infinitely often and $\red$ only finitely often.

\begin{theorem}\label{thm.correctnessNPtoDR}
	Let $\bbA$ be a parity automaton and ${\bbA}^D$ the deterministic Rabin automaton defined from $\bbA$. Then $L(\bbA) = L({\bbA}^D)$.
\end{theorem}

\begin{remark}\label{rem.complexityNPtoDR}
For a parity automaton $\bbA$ of size $n$ with highest even priority $m$, 
our construction produces a deterministic Rabin automaton with
$n^{\calO(m\cdot n)}$ macrostates and $\calO(m\cdot 2^n)$ Rabin pairs, see Lemma \ref{lem.complexNPtoDR} in the appendix.
\end{remark}


%% file: sec.ProofSystem.tex
\section{\BT proofs}\label{sec.BTproofs}

\subsection{Proof systems}
We present two non-wellfounded proof systems for the modal $\mu$-calculus, namely \BT and \BTInf. The idea is that annotated sequents in the \BT system correspond to macrostates of $\bbA^D$, where $\bbA$ is a nondeterministic parity automaton checking the trace condition in an \NW proof. The rules of \BT resemble the transition function of $\bbA^D$.

Let $\Phi$ be a set of formulas, the sequent we want to
prove, and let $m= \max_\Omega(\Phi)$ be the maximal even priority of $\Omega$. \emph{Annotated sequents} are sets
of pairs $(\varphi,\sigma)$, usually written as $\varphi^{\sigma}$, where $\varphi \in \Clos(\Phi)$ and $\sigma \in \TSeq(m)$. For an annotated sequent $\Gamma$ we let $\Gamma^N$ be the set of annotations occurring in $\Gamma$, i.e. $\Gamma^N = \{\sigma \in \TSeq(m) \| \exists \varphi \text{ s.t. } \varphi^{\sigma} \in \Gamma\}$. We let $\Gamma_k^N$ be the set of binary strings occurring at the $k$-th position of the annotations in $\Gamma$, i.e., $\Gamma_k^N = \pi_k[\Gamma^N]$. We say that a string $s$ \emph{occurs in $\Gamma_k^N$} if there exists $t \in \Gamma_k^N$ such that $s \sqsubseteq t$.

 For $\sigma = (s_0,...,s_m)\in \TSeq(m)$ we define $\sigma\cdot 1_k = (s_0,...,s_k1,...,s_m)$ and $\sigma\cdot 0_k = (s_0,...,s_k0,...,s_m)$. For an annotated sequent $\Gamma$ we let $\Gamma^{\cdot 0_k}$ denote the annotated sequent $\{\varphi^{\sigma\cdot 0_k} \| \varphi^{\sigma} \in \Gamma\}$.

Let $\Gamma$ be an annotated sequent and $\varphi^{\sigma} \in \Gamma$.
We define $\sigma\upto k^{\Gamma}$ to be the tuple of binary strings
obtained from $\sigma = (s_0,...,s_m)$ by replacing every $s_j$ with $j
> k$ by $\minL(\tree(\Gamma_j^N)$. If the context $\Gamma$ is clear we write $\sigma\upto k$ instead of $\sigma\upto k^{\Gamma}$.

\begin{figure}[htb]
	\begin{minipage}{\textwidth}
		\begin{minipage}{0.22\textwidth}
			\begin{prooftree}
				\infer[left label=\AxLit:]0{ p^{\sigma},\bar{p}^{\tau}, \Gamma}
			\end{prooftree}
		\end{minipage}
		\begin{minipage}{0.18\textwidth}
			\begin{prooftree}
				\infer[left label=\AxTop:]0{ \top^\sigma, \Gamma}
			\end{prooftree}
		\end{minipage}
		\begin{minipage}{0.25\textwidth}
			\begin{prooftree}
				\hypo{ \varphi^{\sigma},\psi^{\sigma}, \Gamma}
				\infer[left label= \RuOr:]1{ (\varphi \lor \psi)^{\sigma}, \Gamma}
			\end{prooftree}
		\end{minipage}
		\begin{minipage}{0.27\textwidth}
			\begin{prooftree}
				\hypo{ \varphi^{\sigma},\Gamma\quad \psi^{\sigma}, \Gamma}
				\infer[left label= \RuAnd:]1{ (\varphi \land \psi)^{\sigma}, \Gamma}
			\end{prooftree}
		\end{minipage}
	\end{minipage}
	
	\bigskip
	
	\begin{minipage}{\textwidth}
		
		\begin{minipage}{0.40\textwidth}
			\begin{prooftree}
				\hypo{ \varphi\unsubst{x}{\mu x.\varphi}^{\sigma\upto \Omega(\mu x.\varphi)},\Gamma}
				\infer[left label=\RuMu:]1{ \mu x.\varphi^{\sigma}, \Gamma}
			\end{prooftree}
		\end{minipage}
		\begin{minipage}{0.30\textwidth}
			\begin{prooftree}
				\hypo{ \varphi\unsubst{x}{\nu x.\varphi}^{\sigma\upto k \cdot 1_{k}},\Gamma^{\cdot 	0_{k}}}
				\infer[left label=\RuNu:]1[~ where $k = \Omega(\nu x.\varphi)$]{ \nu x.\varphi^{\sigma}, \Gamma}
			\end{prooftree}
		\end{minipage}
	\end{minipage}
	
	\bigskip

	\begin{minipage}{\textwidth}
	\begin{minipage}{0.30\textwidth}
		\begin{prooftree}
			\hypo{ \varphi^{\sigma},\Gamma}
			\infer[left label= \RuBox:]1[]{ \Box \varphi^{\sigma}, \Diamond \Gamma, \Delta}
		\end{prooftree}
	\end{minipage}
		\begin{minipage}{0.50\textwidth}
			\begin{prooftree}
				\hypo{ \varphi^{\sigma}, \Gamma}
				\infer[left label=\RuResolve:]1[~ where $\sigma > \tau$]{ \varphi^{\sigma}, \varphi^{\tau},\Gamma}
			\end{prooftree}
		\end{minipage}
		\begin{minipage}{0.15\textwidth}
			\begin{prooftree}
				\hypo{~~[\Gamma]^{\dx}}
				\infer[no rule]1{\vdots}
				\infer[no rule]1{\Gamma}
				\infer[left label=\RuDischarge:]1{\Gamma}
			\end{prooftree}
		\end{minipage}
	\end{minipage}

	\bigskip
	
	\begin{minipage}{\textwidth}
		\begin{minipage}{0.20\textwidth}
			\begin{prooftree}
				\hypo{ \varphi_1^{(...,st_1,...)},...,\varphi_n^{(...,st_n,...)}, \Gamma}
				\infer[left label=\RuCompress[k][s0]:]1[~ where $s$ does not occur in $\Gamma_k^N$]{ 	\varphi_1^{(...,s0t_1,...)},...,\varphi_n^{(...,s0t_n,...)}, \Gamma}
			\end{prooftree}
		\end{minipage}
	\end{minipage}
	
	\bigskip
	
	\begin{minipage}{\textwidth}
		\begin{minipage}{0.20\textwidth}
			\begin{prooftree}
				\hypo{ \varphi_1^{(...,st_1,...)},...,\varphi_n^{(...,st_n,...)}, \Gamma}
				\infer[left label=\RuCompress[k][s1]:]1[~ where $s$ does not occur in $\Gamma_k^N$ and $s\neq 	0\cdots0$]{ \varphi_1^{(...,s1t_1,...)},...,\varphi_n^{(...,s1t_n,...)}, \Gamma}
			\end{prooftree}
		\end{minipage}
	\end{minipage}
	
	\caption{Rules of \BT}
	\label{fig.BTrules}
\end{figure}

The rules \RuCompress[k][s0] and
\RuCompress[k][s1] are schemata for $k = 0,2,...,m$ and $s \in 2^*$. In
these rules the notation $\varphi_i^{(\dots,s t_i,\dots)}$ is to
be read such that $s t_i$ is the binary string in the $k$-th position
of the annotation. We will write \RuCompress for any of those rules and write \RuCompress[k][s] for either \RuCompress[k][s0] or \RuCompress[k][s1]. 

Note that, if one ignores the annotations, the rules \AxLit, \AxTop,
\RuOr, \RuAnd, \RuMu, \RuNu and \RuBox in Figure \ref{fig.BTrules} are
the same as the rules of \NW. 
As mentioned above annotated sequents in the \BT system correspond to macrostates of $\bbA^D$, where $\bbA$ is a nondeterministic parity automaton checking the trace condition in an \NW proof. The rules of \BT correspond to the transition function $\delta_A$ of $\bbA^D$, where the transformations of $\delta_A$ are distributed over multiple rules:
Step 1(a) of $\delta_A$ is carried out in every rule and step 1(b) and step 2 correspond to the modification of the annotations in the rules R$_\mu$ and R$_\nu$. Notably, we do not add zeros to the annotations if the zeros would get deleted anyway in step 4 of the transition function.
The rules \RuResolve and \RuCompress are
additional and correspond to steps 3 and 4 of
$\delta_A$.

\begin{definition} 
A \BT \emph{derivation} $\pi$ is a derivation defined from the rules in
Figure \ref{fig.BTrules}, such that the rules are applied with the following
priority: first \RuResolve, then \RuCompress, and then all other rules.		
\end{definition}

Just as annotated sequents correspond to macrostates of the deterministic automaton $\bbA^D$, the soundness condition of \BTInf and \BT correspond to the acceptance condition of $\bbA^D$: We say that a pair $(k,s)$ is preserved at a node, if $s$ is in play at position $k$ at the corresponding macrostate and not marked red; and progresses if it is marked green. 
	
\begin{definition}
Let $\pi$ be a \BT derivation of $\Phi$, $m = \max_\Omega(\Phi)$ and $S$ be a set of nodes in $\pi$. 
Let $k \in \{0,2,...,m\}$ and $s \in 2^*$. We say that the pair $(k,s)$ 
		\begin{itemize}
			\item is \emph{preserved} on $S$ if 
			\begin{itemize}
				\item $s$ occurs in $\sfS(v)_k^N$ for every $v$ in $S$ and
				\item if $\sfR(v) =  \RuCompress[k][t]$ for a node $v$ in $S$, then $t \not\sqsubset s$,
			\end{itemize}
			\item \emph{progresses} (infinitely often) on $S$ if there is $s'= s0\cdots0$ such that $\sfR(v) = \RuCompress[k][s'1]$ for some $v$ in $S$ (for infinitely many $v \in S$). 
		\end{itemize}
	\end{definition}

\begin{definition}
Let $\pi$ be a \BT derivation. 
An infinite branch $\alpha = (u_i)_{i\in \omega}$ in $\pi$ is \emph{successful} 
if there are $N$ and $(k,s)$ such that $(k,s)$ is preserved and progresses 
infinitely often on $\{u_i \| i \geq N\}$.
A \emph{$\BTInf$ proof} is a \BT derivation without occurrences of \RuDischarge 
and such that all infinite branches are successful.
A \emph{$\BT$ proof} is a finite \BT derivation such that for each \emph{strongly
connected subgraph} $S$ in $\calT_\pi^C$ there exists $(k,s)$ that is preserved
and progresses on $S$.

We write $\BT \vdash \Gamma$ ($\BTInf \vdash \Gamma$) if there is a \BT (\BTInf) proof of $\Gamma$, i.e., a proof, where $\Gamma$ is the sequent at the root of the proof. 
\end{definition}
	
\begin{remark}
\label{r:jsbt}
In the proof system \JS introduced by Jungteerapanich and Stirling \cite{Jungteerapanich2010,Stirling2014} annotated sequents are of the form $\theta \vdash \phi_1^{a_1},...,\phi_n^{a_n}$, where $a_1,...,a_n$ are sequences of names and the so-called \emph{control} $\theta$ is a linear order on all names occurring in the sequent.
In contrast to \JS our sequents consist of formulas with annotations and 
nothing else, that is, no control.
On the other hand the soundness condition of \BT is less local: It speaks about
strongly connected subgraphs, whereas in \JS only paths between leafs and its companions have to be checked.
We see that the control in \JS gives information on the structure of the cyclic 
proof tree. 
Interestingly, we could also add a control to our sequents and obtain a 
soundness condition that talks about paths, if desired. 
Similarly, in \cite{Afshari2022} a control was added to a cyclic system for the
first-order $\mu$-calculus introduced by \cite{Sprenger03} to obtain a path-based
system.
\end{remark}

\subsection{Soundness and Completeness}\label{sec.sub.soundCompleteness}
The intuitive idea behind the \BTInf proof system is the following: 
Starting with an \NW proof, we can define a nondeterministic parity automaton
$\bbA$, that checks if an infinite branch carries a $\nu$-trail. 
Using the determinization method from Section \ref{sec.det} we simulate 
macrostates of $\bbA^D$ by annotated formulas in the proof system. 
Thus an infinite branch in \BTInf resembles an infinite run of $\bbA^D$.
This will be formalised in the Soundness and Completeness proofs.

\paragraph{Tracking automaton}

Let $\Phi$ be a sequent of formulas, $\eta x_1. \psi_1,...,\eta x_n. \psi_n$ the fixpoint formulas in $\Fix(\Phi)$ and $\Omega$ the parity function on $\Fix(\Phi)$.

We define a nondeterministic parity automaton that checks if there is a 
$\nu$-trail on an infinite branch of some \NW proof of $\Phi$. 
The alphabet $\Sigma$ consists of all triples $(\Gamma,\xi,\Gamma')$, where 
$\Gamma \subseteq \Clos(\Phi)$ is the conclusion and $\Gamma' \subseteq 
\Clos(\Phi)$ is the premise of a rule in Figure \ref{fig.NWrules} with principal
formula $\xi$. 
We define the following nondeterministic parity automaton $\mathbb{A} = (A,\Delta, a_I,\Omega_A)$:
\begin{itemize}
	\item $A = a_I \cup \Clos(\Phi) \cup \{\eta x. \psi^* \| \eta x. \psi \in \Clos(\Phi)\}$,
	\item For each $\gamma \in A $ and $(\Gamma,\xi,\Gamma') \in \Sigma$:
	\begin{enumerate}
		\item if $\gamma = a_I$, then $\Delta(\gamma,(\Gamma,\xi,\Gamma')) = \Phi$,
		\item if $\gamma = \xi = \eta x.\psi$ then $\Delta(\gamma,(\Gamma,\xi,\Gamma')) = \{\eta x. \psi^*\}$,
		\item if $\gamma = \eta x. \psi^*$, then $\Delta(\gamma,(\Gamma,\xi,\Gamma')) = \{\gamma' \| (\psi\unsubst{x}{\eta x. \psi},\gamma') \in \sfT_{\Gamma,\xi,\Gamma'} \}$ and
		\item else $\Delta(\gamma,(\Gamma,\xi,\Gamma')) = \{\gamma' \| (\gamma,\gamma') \in \sfT_{\Gamma,\xi,\Gamma'}\}$.
	\end{enumerate}
\item For all states $\eta x. \psi^*$ let $\Omega_A(\eta x. \psi^*) = \Omega(\eta x. \psi)$. For all other states $a$ let $\Omega_A(a)= \max_\Omega(\Phi)$ if $\max_\Omega(\Phi)$ is odd and $\Omega_A(a)= \max_\Omega(\Phi) +1$ else.
\end{itemize}

Let $\alpha = (v_n)_{n\in\omega}$ be an infinite branch in an \NW-proof $\pi$. 
We define $w(\alpha) \in \Sigma^{\omega}$ to be the infinite word $(\sfS(v_0),\sff(v_0),\sfS(v_0))(\sfS(v_0),\sff(v_0),\sfS(v_1))(\sfS(v_1),\sff(v_1),\sfS(v_2))...$.

\begin{lemma}\label{lem.trackAut}
	Let $\alpha$ be an infinite branch in an \NW proof. Then $\alpha$ carries a $\nu$-trail iff $w(\alpha) \in \calL(\bbA)$.
\end{lemma}

Combining Lemma \ref{lem.trackAut} and Theorem \ref{thm.correctnessNPtoDR} from Section \ref{sec.det}  we get
\begin{lemma}\label{lem.NWiffDRautomaton}
	Let $\pi$ be an \NW derivation. Then $\pi$ is an \NW proof iff for every infinite branch $\alpha$ in $\pi$ it holds $w(\alpha) \in \calL(\bbA^D)$.
\end{lemma}

\begin{lemma}\label{lem.NWiffBT}
	Let $\Gamma$ be a sequent. Then $\NW \proves \Gamma$ iff $\BT \proves \Gamma^{\epsilon}$.
\end{lemma}
\begin{proof}[Sketch]
Let $\pi$ be an \NW proof of a sequent $\Gamma$. 
Inductively we translate every node $v$ in $\pi$ to a node $v'$ plus some 
additional nodes, such that $v'$ is labeled by the same sequent as $v$ plus 
annotations. 
This can be achieved by replacing every rule in \NW by its corresponding rule in 
\BT and adding the rules \RuResolve and \RuCompress whenever applicable. 
This yields a \BT derivation $\rho$. 
It remains to show that every infinite branch $\alpha= (v_i)_{i\in\omega}$ in 
$\rho$ is successful. 
Let $\hat{\alpha}$ be the corresponding infinite branch in $\pi$. 
Due to Lemma \ref{lem.NWiffDRautomaton} it holds that $\hat{\alpha} \in 
\calL(\bbA^D)$. 
Thus there is $(k,s)$ such that $s$ is in play at position $k$ cofinitely often 
and $c_k(s)$ is \green infinitely often and \red only finitely often. 
As the annotations in $\alpha$ resemble the annotations in the run of $\bbA^D$ 
on $\hat{\alpha}$ it follows that there is some $N\in\omega$ such that $(k,s)$
is preserved and progresses infinitely often on $\{v_i\| i\geq N\}$.
	
Conversely let $\rho$ be a \BT proof of $\Gamma^{\epsilon}$. 
We let $\pi$ be the \NW derivation defined from $\rho$ by omitting the rules 
\RuResolve and \RuCompress and reducing the other rules to the corresponding 
\NW rules. 
We have to show that every infinite branch $\alpha$ in $\pi$ is successful. 
Let $\alpha'= (v_i)_{i\in\omega}$ be the corresponding infinite branch in $\rho$. 
Because $\rho$ is a \BT proof there is $N$,$(k,s)$ such that $(k,s)$ is preserved
and progresses infinitely often on $\{v_i\| i\geq N\}$. 
Again the annotations in $\alpha'$ resemble the annotations in the run of 
$\bbA^D$ on $\alpha$, thus $(k,s)$ witnesses the acceptance of the run of 
$\calL(\bbA^D)$ on $\alpha$ and Lemma \ref{lem.NWiffDRautomaton} concludes the
proof.
\end{proof}

\begin{theorem}[Soundness and Completeness]\label{thm.BTInfSoundnessCompleteness}
Let $\Gamma$ be a sequent. 
Then there is a $\BTInf$-proof of $\Gamma^{\epsilon}$ iff $\bigvee \Gamma$ is 
valid.
\end{theorem}
\begin{proof}
	This follows from Lemma \ref{lem.NWiffBT} and Theorem \ref{thm.NWSoundCompleteness}.
\end{proof}

\subsection{Cyclic \BT proofs}\label{sec.sub.cyclicBT}

As \NW proofs can be assumed to be regular and annotations are added deterministically we can also assume \BTInf proofs to be regular. A standard argument then transforms regular \BTInf proofs into \BT proofs and vice versa.

\begin{lemma}\label{lem.BTInfToBT}
	An annotated sequent is provable in $\BT$ iff it is provable in $\BTInf$.
\end{lemma}

\begin{theorem}[Soundness and Completeness]
	Let $\Gamma$ be a sequent. Then there is a $\BT$-proof of $\Gamma^{\epsilon}$ iff $\bigvee \Gamma$ is valid..
\end{theorem}

\begin{remark}
The number of distinct subtrees in a regular \BTInf proof can be bounded by the
number of distinct annotated sequents. This follows because
the same statement holds for \NW proofs \cite{Niwinski1996} and because
in the proof of Lemma \ref{lem.NWiffBT} annotations and extra rules are
added deterministically to sequents in \NW proofs. 
	
Let $\Phi$ be a sequent, $n = |\Clos(\Phi)|$ and $m = \max_\Omega(\Phi)$. 
There are at most $n^{\calO(m\cdot n)}$ many distinct annotated sequents
occurring in a proof of $\Phi$, because annotated sequents resemble macrostates
in $\bbA^D$ and as seen in  Remark \ref{rem.complexityNPtoDR} there are at most 
$n^{\calO(m\cdot n)}$ distinct macrostates in $\bbA^D$.
	
Combining these two observations with the proof of Lemma \ref{lem.BTInfToBT}
yields that the height of a \BT proof of a sequent $\Phi$ can be bound by 
$n^{\calO(m\cdot n)}$. 
This is the same complexity as in \JS \cite{Jungteerapanich2010}. 
\end{remark}

\begin{remark}
Given a \BT derivation $\pi$, we can check if $\pi$ is a \BT proof in
$\mathrm{coNP}$. 
We can give the following algorithm in $\mathrm{NP}$, that checks if $\pi$ is 
not a \BT proof: Choose non-deterministically a strongly connected subgraph $S$
and check if there exists $(k,s)$ that is preserved and progresses on $S$, the
latter can be done in polynomial time. 
The complexity of proof checking can be compared to linear time in \JS and $\mathrm{PSPACE}$ in \NW. Note that, if we add a control to the \BT proof system, the soundness condition boils down to checking paths between leafs and its companions. In that case proof checking could also be done in linear time.
\end{remark}

%% file: sec.Conclusion.tex
\section{Conclusions and Future Work}

We hope that this paper contributes to the theory of non-wellfounded and cyclic 
proof systems by discussing applications of automata theory in the field.
We have argued for the relevance of the notion of determinizing 
stream automata in the design of proof systems for the modal $\mu$-calculus.
More concretely, we have introduced a determinization construction based on binary
trees and used this to obtain a new derivation system \BT which is cyclic, 
cutfree, and sound and complete for the collection of valid $\muML$-formulas.
In the remainder of this concluding section we point out some directions for 
future research.

First of all, our approach is not restricted to the modal $\mu$-calculus, but
will apply to non-wellfounded and cyclic derivation systems for many other logics 
as well.
For instance, in the proof systems  $\mathrm{LKID}^\omega$ \cite{Brotherston2006}
for first-order logic with inductive definitions, cyclic arithmetic $\mathrm{CA}$
\cite{Simpson2017} and similar systems the trace condition is of the form that 
on every infinite branch there is a term/variable which progresses infinitely 
often.
This condition can be checked by a nondeterministic Büchi automaton and thus our method would yield an annotated proof system, where the annotations are binary strings, 
which label the terms/variables. 

Second, in Remark~\ref{r:jsbt} we discussed some relative advantages and 
disadvantages of the systems \JS and \BT.
It would be interesting to either design a system that combines the advantages
of both systems (i.e. sequents consisting of annotated formulas only as in \BT, 
and a local condition for proof checking as in \JS), or prove that such a system 
cannot exist.


Finally, it would be interesting (and in fact, it was one of the original aims 
of our work), to connect annotation-based sequent calculi such as \JS and \BT to
Kozen's Hilbert-style proof system and to see whether a more 
structured automata-theoretic approach would yield an alternative proof
of Walukiewicz' completeness result.
Note that this was also the goal of Afshari \& Leigh \cite{Afshari2017}; 
unfortunately, it was recently shown by the second author \cite{Kloibhofer2023} that
the system $\mathsf{Clo}$, a key system in Afshari \& Leigh's approach linking
\JS to Kozen's axiomatization, is in fact incomplete.

%% file: sec.Appendix.tex
\section{Appendix}\label{sec.appendix}

\subsection{Determinization of automata with binary trees}

\begin{proposition}\label{prop.deltaProp}
	\begin{enumerate}
		\item 	The transition function $\delta$ is
		well-defined. In particular, the result of step 4 does not depend on the
		order in which the witnesses $t \in T$ are chosen.\label{prop.deltaProp.wellDefined}
		\item The string $\epsilon$ is never marked red. \label{prop.deltaProp.epsilon}
		\item The lengths of the binary strings occurring in $\ran(f)$ are bounded by the size of $B$.\label{prop.deltaProp.Size}
	\end{enumerate}	
\end{proposition}
\begin{proof}
\ref{prop.deltaProp.wellDefined}. After steps 1-3 of the transition function $S'$ consists of a finite set of binary strings such that for all $s \neq t$ in $S'$ it holds $s\not\sqsubseteq t$. Thus $T^{S'}$ describes a tree, where every node has at most two children and the leaves are labeled by disjoint sets of states. Step 4 of $\delta$ identifies a node $t$ with its child $t'$, if $t'$ is its unique child. This results in a unique binary tree $T'$. It remains to show that the colouring of $T'$ does not depend on which nodes are identified first.
Therefore we will give an equivalent presentation of step 4 of $\delta$. Let $\sim$ be the equivalence relation on $T^{S'}$ generated by all pairs of nodes $(s,t)$ such that $t$ is the unique child of $s$. Let $T^E$ be the quotient of $T^{S'}$ over $\sim$, i.e. $T^E = \{[s]_{\sim} \| s \in T^{S'}\}$. We can define a parent relation on $T^E$ as follows: $[s]_\sim$ is the parent of $[t]_\sim$ iff $s \not\sim t$ and there is $s'\sim s$ and $t'\sim t$ such that $s'$ is the parent of $t'$. If $[s]_\sim$ and $[t]_\sim$ are siblings in $T^E$, then there are $s' \sim s$ and $t' \sim t$ such that $s'$ and $t'$ are siblings in $T^{S'}$, hence they inherit an order on the siblings. Thus $T^E$ is a binary tree, which can be given as a set of binary strings. We can define a colouring map $c$ on $T^E$ as follows: $[s]_\sim$ is coloured red if it has an ancestor $[t]_\sim$ with $|[t]_\sim| > 1$. An equivalence class $[s]_\sim$ is coloured green if it is not red and it has a minimal descendant (i.e. every child in the ancestor path is the minimal child wrt. the sibling order) $[t]_\sim$ such that there are $t' \sim t'' \sim t$ where more $1$s are occurring in $t'$ than in $t''$. All other nodes are marked white. It can be easily seen that the coloured binary tree $T^E$ is isomorphic to $T^{S'}$ after step 4 independent of which nodes are identified first.

\ref{prop.deltaProp.epsilon}. The string $\epsilon$ is not a real superstring of any other string, thus it can never be marked red.
\ref{prop.deltaProp.Size} follows, as the length of a path from the root to a leaf is bounded by the size of the tree.
\end{proof}

\setcounter{theorem}{\getrefnumber{thm.CorrectnessNBtoDR}}
\addtocounter{theorem}{-1}
\begin{theorem}
	$\bbB$ accepts a word $w$ iff $\bbB^D$ accepts $w$.
\end{theorem}
\begin{proof}
	Let $w = y_0y_1... \in \Sigma^{\omega}$ and $\rho^S = S_0S_1S_2...$ be the run of $\bbB^D$ on $w$. 
	
	Suppose there is an accepting run $\rho = b_0b_1b_2..$ of $\bbB$ on $w$. Let $s_n$ be the binary string such that $b_n^{s_n} \in S_n$ for $n \in \omega$. Let $t$ be the maximal string which is a substring of cofinitely many $s_n$ and only marked red finitely often. 
	Note that such a string exists, as $\epsilon$ always satisfies these conditions due to Proposition \ref{prop.deltaProp}.\ref{prop.deltaProp.epsilon}. By definition $t$ is in play cofinitely often and red only finitely often, we will show that $c(t)$ is green infinitely often. Let $t'$ be the maximal string of the form $t0\cdots0$, such that $t' \sqsubseteq s_n$ for cofinitely many $n$. Let $N$ be such that $t'\sqsubseteq s_n$ and $c(t) \neq \red$ in $S_n$ for all $n > N$. 
	
	Now we distinguish the following cases: $t'=s_n$ for infinitely many $n$ (first case), $t'1 \sqsubseteq s_n$ for cofinitely many $n$ (second case) and $t'1 \sqsubseteq s_n$ for infinitely many $n$ and $t'0 \sqsubseteq s_n$ for infinitely many $n$, while $t' = s_n$ for only finitely many $n$ (third case). These three cases cover all possibilities, as $t'$ is the maximal string of the form $t0\cdots0$, such that $t' \sqsubseteq s_n$ for cofinitely many $n$. 
	
	First assume that $s_n = t'$ for infinitely many $n$. Let $m > N$, such that $s_{m} = t'$. As $\rho$ is an accepting run, $1$ will be appended to $s_n$ at step 2 of $\delta$ for some  $n > m$. This $1$ will need to be removed at some possibly later stage, as $s_n = t'$ infinitely often. But the $1$ can only get removed in step 4(b), which means that $t$ is marked green, as $c(t)$ is never red. Thus $c(t) = \green$ infinitely often.
	
	Secondly let $t'1\sqsubseteq s_n$ for cofinitely many $n$. Let $N_1 > N$, such that $t'1\sqsubseteq s_n$ for all $n > N_1$. The definition of $t$ implies that $t'1$ is marked red infinitely often. This can only happen if there is a witness $r \sqsubseteq t'$ in step 4 of $\delta$. As $t$ is never marked red it follows that $t \sqsubseteq r$. If $t \sqsubseteq r \sqsubset t'$, then we are in step 4(a) and $s_n$ is replaced by $s_n\unsubst{r0}{r}$. Yet $t'\not\sqsubseteq s_n\unsubst{r0}{r}$, which contradicts our assumption. Thus the witness in step 4 of $\delta$ has to be $t'$. In this case $t'1 \in T^{S_n}$ and $t'0 \notin T^{S_n}$, thus $t$ is marked green.
	
	Thirdly consider the case where $t'0 \sqsubseteq s_n$ for infinitely many $n$ and $t'1 \sqsubseteq s_n$ for infinitely many $n$, while $s_n = t'$ for only finitely many $n$. Then it holds for infinitely many $n > N$, that $t'1 \sqsubseteq s_n$ and $t'0 \sqsubseteq s_{n+1}$. As $s_{n+1} < s_n$, this is only possible if $t'$ is the witness in step 4(b) of the transition function, i.e. $t'1 \in T$ and $t'0 \notin T$, which means that $t$ is marked green as it is never marked red.

	Thus in every case $t$ is marked green infinitely often and the first direction is proven.

	\bigskip

	Conversely, suppose that there is a binary string $t$ which is in play cofinitely often and which is labeled green infinitely often and red only finitely often. Let $N$ be such that $t$ is in play and never marked red for any $n \geq N$. For $i \geq N$ we define 
	\[
	A_i = \{b^s \in S_i ~|~ t\sqsubseteq s\}.
	\]
	We first show
	\begin{align}
		&\text{For all } b^s \in A_N\text{ there is a path from }b_I\text{ to }b\text{ in }\bbB\text{ on input } y_0\cdots y_{N-1}.\label{stat.binProofInitial}\\
		&\text{For all }i > N\text{ and }b^{s_b} \in A_{i+1}\text{ there exists }a^{s_a} \in A_i\text{ such that }b \in \Delta(a,y_i)\label{stat.binProofSucc}
	\end{align}
	Statement (\ref{stat.binProofInitial}) follows, as the transition function is just a refined version of a macro-move. For (\ref{stat.binProofSucc}) let $b^{s_b} \in A_{i+1}$. Due to step 1 of the transition function there is $a^{s_a} \in S_i$ with $b \in \Delta(a,y_i)$. We choose $a^{s_a}$ with that property such that $s_a$ is maximal and claim that $t \sqsubseteq s_a$. To see that we take a look at $\delta(S_i,y)$. After step 2 there is $b^{s_a0}$ or $b^{s_a1}$ in $S'_{i}$, which will not be removed in step 3 as we chose the maximal $s_a$. Now let $r$ be maximal such that $r \sqsubseteq s_a$ and $r \sqsubseteq t$. Note that $s_a \sqsubset t$ can not hold, because this would imply $A_i = \varnothing$. If $r=t$ we are done. Else $r$ has to be a witness in step 4 of the transition function, as $t \sqsubseteq s_b$. As $s_a \in T^{S'_i}$ this implies that $t \notin T^{S'_{i}}$. Yet $t \in T^{S_{i+1}}$, but this is only possible if $c(t) = \red$ at the end of step 4, which contradicts our definition of $t$.
	
	\medskip
	
	We are now able to define the trace tree $\calT^\bbB$. It will consist of the root $b_I$ and nodes $(a^{s_a},i)$, where $a^{s_a} \in A_i$ for $i \geq N$. We define a partial order $<_{\calT}$ on the nodes of $\calT^{\bbB}$ in the following way: $(a^{s_a},i) <_{\calT} (b^{s_b},i)$ iff $s_a < s_b$. The parent of $(a^{s_a},N)$ is $b_I$. For $i \geq N$ the unique parent of $(b^{s_b},i+1)$ is an element $(a^{s_a},i)$ such that $b \in \Delta(a,y_i)$ which is maximal with respect to $<_{\calT}$. Such an element always exists due to (\ref{stat.binProofSucc}), if there exist more than one, choose one of them. 
	
	The trace tree $\calT^\bbB$ is an infinite, finitely branching tree. So by König's Lemma there exists an infinite branch $b_I(a^{s_N}_N,N)(a^{s_{N+1}}_{N+1},N+1)...$. We let $\rho^{\mathcal{T}}$ be the infinite branch such that the infinite string $s_Ns_{N+1}...$ is minimal with respect to the lexicographical order. Due to (\ref{stat.binProofInitial}) there exists a path $\rho'$ from $b_I$ to $a_N$ in $\bbB$ on input $y_0...y_{N-1}$. Combining that with (\ref{stat.binProofSucc}) we obtain that $\rho = \rho'a_Na_{N+1}...$ is a run of $\bbB$ on $w = y_0y_1...$.

	It remains to show that $\rho$ is successful, i.e. that $a_j \in F$ for infinitely many $j \geq N$. To do so, assume that there is $m \geq N$, such that $a_{i} \notin F$ for all $i \geq m$. We have $a_i^{s_i} \in A_i$ for $i \geq m$, where $s_m = tr_m$ for some $r_m$. Let $r_i$ be minimal such that $s_i = tr_i0\cdots0$ for $i > m$. As $a_i \notin F$ for all $i \geq m$, only zeros are added to $s_i$, hence the length of $r_i$ is decreasing. Now $\rho^\calT$ is the minimal infinite path in $\calT^\bbB$, thus all paths left of $\rho^{\calT}$ are finite, which implies that at some point $k \geq m$ it holds $s_k = t0\cdots0$. Now the next time when $t$ gets labeled green, $t0\cdots0$ has to be the witness in step 4(b), which is only possible if $a_j \in F$ for some $j \geq k$.
	
\end{proof}

\begin{lemma}\label{lem.complexityNBtoDR}
	For a Büchi automaton with $n$ states, the automaton $\bbB^D$ has $n^{\calO(n)}$ macrostates and the Rabin condition consists of $\calO(2^n)$ pairs.
\end{lemma}
\begin{proof}
	The number of binary trees with $k+1$ leaves $C_k$ is called the $k$-th Catalan number. It holds that $C_k = \frac{1}{k+1} \binom{2k}{k} \leq 2^{2k}$. A binary tree with $k$ leaves has $2k-1$ nodes, thus there are $3^{2k-1}$ possible colouring maps $c$. The $k$ leaves can be labeled by disjoint subsets of $n$ in $(k+1)^n$ different ways, in particular there are at most $(k+1)^n$ possibilities to label $k$ leaves with disjoint nonempty subsets of $n$. A macrostate can be represented by a labeled, coloured binary tree of size $k$, where $k$ ranges from $0,...,n$. In total we have
	\begin{align*}
		|B^D| &\leq \sum_{k=0}^{n} C_{k-1}\cdot 3^{2n-1} \cdot(k+1)^n\\
		&\leq 6^{2n-1}\cdot (n+1)^{n+1} = n^{\calO(n)}.
	\end{align*}
	The number of Rabin pairs is the number of binary strings, which may occur in a macrostate. Due to Proposition \ref{prop.deltaProp}.\ref{prop.deltaProp.Size} this is $\sum_{k=0}^n 2^k \leq 2^{n+1}$.
\end{proof}

\setcounter{theorem}{\getrefnumber{thm.correctnessNPtoDR}}
\addtocounter{theorem}{-1}
\begin{theorem}
	Let $\bbA$ be a parity automaton and ${\bbA}^D$ the deterministic Rabin automaton defined from $\bbA$. Then $L(\bbA) = L({\bbA}^D)$.
\end{theorem}
\begin{proof}
	This proof follows the same lines as the proof of Theorem \ref{thm.CorrectnessNBtoDR}. We will focus on the differences in the case of parity automata and omit some steps which are analogous to the Büchi case.

	Let $w = y_0y_1... \in \Sigma^{\omega}$ and $\rho^S = S_0S_1S_2...$ be the run of $\tilde{\bbA}^D$ on $w$. 
	
	Suppose there is an accepting run $\rho = a_0a_1a_2..$ of $\bbA$ on $w$ and let $k \in \{0,...,m\}$ be minimal such that $\Omega(a_j) = k$ for infinitely many $j \in \omega$. Because $\rho$ is accepting, $k$ is even. Let $s_n$ be the binary string such that $a_n^{\sigma_n} \in S_n$ and $s_n = (\sigma_n)_k$ for $n \in \omega$. Let $t$ be the maximal string which is a substring of cofinitely many $s_n$ and $c_k(t)= \red$ only finitely often. 
	Note that $t \neq 0\cdots0$: Assume that $t = 0\cdots0$. As $\rho$ is an accepting run a $1$ is appended to $s_n$ at some point after the last time $c_k(t) = \red$. This $1$ can never get deleted again, and $c_k(t1)$ is never $\red$ because this would imply $c_k(t)$. Thus $t1$ is also a substring of cofinitely many $s_n$ and $c_k(t1) = \red$ only finitely often.
	
	By definition $t$ is in play at position $k$ cofinitely often and $c_k(t)= \red$ only finitely often, we will show that $c_k(t)$ is green infinitely often. Let $t'$ be the maximal string of the form $t0\cdots0$, such that $t' \sqsubseteq s_n$ for cofinitely many $n$. Let $N$ be such that $t'\sqsubseteq s_n$, $c_k(t) \neq \red$ in $S_n$ and $\Omega(a_n) \geq k$ for all $n > N$. 
	
	Now the first direction can be proven the same way as in the proof of Theorem \ref{thm.CorrectnessNBtoDR}.

	\bigskip

	Conversely, suppose that there is $k \in \{0,2,...,m\}$ and a binary string $t$, which is in play at position $k$ cofinitely often such that $c_k(t)$ is $\green$ infinitely often and $\red$ only finitely often, in particular $t \neq 0\cdots0$. 
	Let $N$ be such that $t$ is in play at position $k$ and $c_k(t) \neq \red$ for all $n \geq N$. For $i \geq N$ we define 
	\[
	A_i = \{a^\sigma \in S_i ~|~ t\sqsubseteq \sigma_k\}.
	\]
	We first show
	\begin{align}
		&\text{For all } i > N\text{ and }a^\sigma \in A_i\text{ it holds } \Omega(a) \geq k.\label{stat.binProofArityParity}\\
		&\text{For all } a^\sigma \in A_N\text{ there is a path from }a_I\text{ to }a\text{ in }\bbA\text{ on input } y_0\cdots y_{N-1}.\label{stat.binProofInitialParity}\\
		&\text{For all }i > N\text{ and }b^{\sigma_b} \in A_{i+1}\text{ there exists }a^{\sigma_a} \in A_i\text{ such that }b \in \Delta(a,y_i)\label{stat.binProofSuccParity}
	\end{align}
	For (\ref{stat.binProofArityParity}) assume that $\Omega(a) < k$. Then $\sigma_k$ is replaced by $\minL(T_k)$ at step 1(b) of the transition function, which results in $\sigma_k = 0\cdots0$ at the end of the transition function. Yet $t \neq 0\cdots0$, hence $t \not\sqsubseteq \sigma_k$.
	Statement (\ref{stat.binProofInitialParity}) follows, as the transition function is just a refined version of a macro-move. 
	For (\ref{stat.binProofSuccParity}) let $b^{\sigma_b} \in A_{i+1}$. Due to step 1 of the transition function there is $a^{\sigma_a} \in S_i$ with $b \in \Delta(a,y_i)$. We choose $a^{\sigma_a}$ with that property such that $\sigma_a$ is maximal, let $s_a = (\sigma_a)_k$ and claim that $t \sqsubseteq s_a$. 
	To see that we take a look at $\delta(S_i,y)$. Due to (\ref{stat.binProofArityParity}) after step 2 there is $b^{\sigma}$ in $S'_{i}$, where $\sigma_k = s_a0$ or $\sigma_k = s_a1$. $\sigma$ will not be removed in step 3 as we chose the maximal $\sigma_a$. Now let $r$ be maximal such that $r \sqsubseteq s_a$ and $r \sqsubseteq t$. Note that $s_a \sqsubset t$ can not hold, because this would imply $A_i = \varnothing$. If $r=t$ we are done. Else $r$ has to be a witness in step 4 of the transition function, as $t \sqsubseteq s_b = (\sigma_b)_k$. Because $s_a \in T^{S'_i}_k$ this implies that $t \notin T^{S'_{i}}_k$. Yet $t \in T^{S_{i+1}}_k$, but this is only possible if $c_k(t) = \red$ at the end of step 4, which contradicts our definition of $t$.
	
	\medskip
	
	Using König's Lemma we obtain a run $\rho = \rho'a_Na_{N+1}$ of $\bbA$ on input $w$ analogously as in the proof of Theorem \ref{thm.CorrectnessNBtoDR}.	
	It remains to show that $\rho$ is successful. It suffices to show that $\Omega(a_j) = k$ for infinitely many $j \geq N$ as $\Omega(a_j) \geq k$ for all $j \geq N$ due to (\ref{stat.binProofArityParity}). This can be done as in the proof of Theorem \ref{thm.CorrectnessNBtoDR}.
\end{proof}

\begin{lemma}\label{lem.complexNPtoDR}
	Let $\bbA$ be a parity automaton of size $n$ and highest even parity $m$ and $\bbA^D$ be the deterministic Rabin automaton defined from $\bbA$. Then $\bbA^D$ has $n^{\calO(m\cdot n)}$ macrostates and $\calO(m\cdot 2^n)$ Rabin pairs.
\end{lemma}
\begin{proof}

	Another way to describe a macrostate in $\bbA^D$ is as a collection of macrostates in the automata $(\bbA \cup \bbA_k)^D$ for $k=0,2,...,m$, where $\bbA \cup \bbA_k$ is the Büchi automaton defined above. Hence Lemma \ref{lem.complexityNBtoDR} yields that there are $n^{\calO(m\cdot n)}$ possible macrostates in $\bbA^D$.
	The number of Rabin pairs is the number of possible binary strings for $k=0,2,...,m$, in total $m \cdot \calO(2^n)$.
\end{proof}

\subsection{\BT proofs}

\begin{definition}
	Let $\pi$ be a \BT proof. The \emph{infinite unfolding} $\pi^*$ of $\pi$ is the \BT pre-proof obtained by recursively replacing every discharged leaf $l$ by the subtree of $\pi$ rooted at the  child node of $c(l)$.
\end{definition}

\begin{lemma}
	There is a \BTInf proof of a sequent $\Gamma$ iff there is a regular \BTInf proof of $\Gamma$.
\end{lemma}
\begin{proof}
	The same statement holds for \NW proofs. In the proof of Lemma \ref{lem.NWiffBT} annotations are added in a deterministic way, thus we end up with a regular \BTInf proof.
\end{proof}	

\setcounter{lemma}{\getrefnumber{lem.BTInfToBT}}
\addtocounter{lemma}{-1}
\begin{lemma}
	An annotated sequent is provable in $\BT$ iff it is provable in $\BTInf$.
\end{lemma}	
\begin{proof}
	Let $\pi$ be a regular $\BTInf$-proof of an annotated sequent $\Phi$. For a node $v \in \pi$ we let $\pi_v$ be the subtree of $\pi$ rooted at $v$. We define the equivalence relation $\sim$ by setting $v \sim u$ if $\pi_v = \pi_u$. As $\pi$ is regular there are only finitely many distinct equivalence classes.
	
	Let $P_\pi$ be the set of infinite paths in $\pi$. We will define functions $f_c, f_l$ from $P_\pi \rightarrow \omega$ that select appropriate positions for discharge rules and corresponding leaves. Let $\alpha = (\alpha_i)_{i \in \omega}$ be in $P_\pi$. As $\alpha$ is successful there exist $t < t'$ such that 
	\begin{itemize}
		\item there exists $(k,s)$ which is preserved and progresses on $\{v_t,...,v_{t'}\}$ and
		\item $v_t \sim v_{t'}$.
	\end{itemize}
	Choose $(t,t')$ minimally with respect to the lexicographic order and define $f_c(\alpha) = t$ and $f_l(\alpha) = t'$. Let $C = \{\alpha_{f_c(\alpha)} \| \alpha \in P_\pi\}$ and $L = \{\alpha_{f_l(\alpha)} \| \alpha \in P_\pi\}$ be the set of companions and leaves of our new proof, respectively.  We let the companion of $l = \alpha_{f^l(\alpha)}$ be $c(l) = \alpha_{f^c(\alpha)}$.
	We define $\pi_L$ to be the subtree of $\pi$ up to the leaves $L$. Then we insert a \RuDischarge rule at every companion node $u \in C$ with discharged assumptions $l \in L$ such that $c(l) = u$. Using König's Lemma we can show that $\pi_L$ is indeed a finite tree. 
	
	\noindent
	It remains to show that
	\begin{itemize}
		\item On every strongly connected subgraph $S$ of $\calT_{\pi_L}^C$ there exists $(k,s)$ such that $(k,s)$ is preserved and progresses on $S$.
	\end{itemize}
	It is easy to see that every infinite path $\alpha$ in $\pi_L$ corresponds to an infinite path $\beta$ in $\pi$.
	Now let $S$ be a strongly connected subgraph $S$ of $\calT_\pi^C$ and $\alpha=(\alpha_i)_{i \in \omega}$ be an infinite path in $\pi_L$ that visits every node in $S$ infinitely often and no other node infinitely often. Let $\beta$ be the corresponding infinite path in $\pi$. As $\beta$ is successful there is $N$ and $(k,s)$ such that $(k,s)$ is preserved and progresses infinitely often on $\{\beta_i \| i \geq N\}$. Yet this yields that $(k,s)$ is preserved and progresses on $S$.

	For the other direction we can show that the infinite unfolding $\pi^*$ of a $\BT$-proof $\pi$ is a $\BTInf$-proof. This follows as every infinite branch in $\pi^*$ can be seen as a path in a strongly connected subgraph of $\pi$, thus the condition for the infinite branches follows from the same condition for every strongly connected subgraph of $\pi$.
\end{proof}

%% file: main.bbl
\begin{thebibliography}{10}
\providecommand{\url}[1]{\texttt{#1}}
\providecommand{\urlprefix}{URL }
\providecommand{\doi}[1]{https://doi.org/#1}

\bibitem{Afshari2022}
Afshari, B., Enqvist, S., Leigh, G.E.: Cyclic proofs for the first-order
  $\mu$-calculus. Logic Journal of the IGPL  (2022).
  \doi{10.1093/jigpal/jzac053}, \url{https://doi.org/10.1093/jigpal/jzac053}

\bibitem{Afshari2017}
Afshari, B., Leigh, G.E.: Cut-free completeness for modal mu-calculus. In:
  Proceedings of the 32nd Annual ACM/IEEE Symposium on Logic in Computer
  Science. LICS '17, IEEE Press, Reykjav\'{\i}k, Iceland (2017)

\bibitem{Brotherston2006}
Brotherston, J.: Sequent calculus proof systems for inductive definitions.
  Ph.D. thesis (Nov 2006), \url{https://era.ed.ac.uk/handle/1842/1458}

\bibitem{Calude17}
Calude, C., Jain, S., Khoussainov, B., Li, W., Stephan, F.: Deciding parity
  games in quasipolynomial time. In: Hatami, H., McKenzie, P., King, V. (eds.)
  Proceedings of the 49th Annual {ACM} {SIGACT} Symposium on Theory of
  Computing, ({STOC} 2017). pp. 252--263 (2017)

\bibitem{Doumane2017}
Doumane, A.: Constructive completeness for the linear-time μ-calculus. In:
  2017 32nd Annual ACM/IEEE Symposium on Logic in Computer Science (LICS). pp.
  1--12 (2017). \doi{10.1109/LICS.2017.8005075}

\bibitem{Emerson99}
Emerson, E., Jutla, C.: The complexity of tree automata and logics of programs.
  {SIAM} Journal of Computing  \textbf{29}(1),  132--158 (1999)

\bibitem{Fogarty2015}
Fogarty, S., Kupferman, O., Vardi, M.Y., Wilke, T.: Profile trees for {B}üchi
  word automata, with application to determinization. Information and
  Computation  \textbf{245},  136--151 (2015)

\bibitem{Fogarty2013}
Fogarty, S., Kupferman, O., Wilke, T., Vardi, M.: Unifying {B}üchi
  complementation constructions. Logical Methods in Computer Science
  \textbf{9}(1) (2013), \url{https://doi.org/10.2168%2Flmcs-9%281%3A13%292013}

\bibitem{Friedmann2013}
Friedmann, O., Lange, M.: Deciding the unguarded modal μ-calculus. Journal of
  Applied Non-Classical Logics  \textbf{23}(4),  353--371 (oct 2013).
  \doi{10.1080/11663081.2013.861181}

\bibitem{Janin1995}
Janin, D., Walukiewicz, I.: Automata for the modal $\mu$-calculus and related
  results. In: Proceedings of the Twentieth International Symposium on
  Mathematical Foundations of Computer Science, MFCS'95. LNCS, vol.~969, pp.
  552--562. Springer (1995)

\bibitem{Janin96}
Janin, D., Walukiewicz, I.: On the expressive completeness of the propositional
  $\mu$-calculus w.r.t.\ monadic second-order logic. In: Proceedings~of the
  Seventh International Conference on Concurrency Theory, CONCUR '96. LNCS,
  vol.~1119, pp. 263--277 (1996)

\bibitem{Jungteerapanich2010}
Jungteerapanich, N.: Tableau systems for the modal $\mu$-calculus. Ph.D.
  thesis, School of Informatics; The University of Edinburgh (2010),
  \url{http://hdl.handle.net/1842/4208}

\bibitem{Kloibhofer2023}
Kloibhofer, J.: A note on the incompleteness of {A}fshari \& {L}eigh's system
  {Clo} (2023), \url{https://doi.org/10.48550/arXiv.2307.06846}

\bibitem{Kozen1983}
Kozen, D.: Results on the propositional $\mu$-calculus. Theoretical Computer
  Science  \textbf{27},  333--354 (1983)

\bibitem{Leigh2023}
Leigh, G.E., Wehr, D.: From {GTC} to {Reset}: {Generating} {Reset} {Proof}
  {Systems} from {Cyclic} {Proof} {Systems}. Tech. rep. (Jan 2023).
  \doi{10.48550/arXiv.2301.07544}, \url{http://arxiv.org/abs/2301.07544}

\bibitem{Loeding2019}
Löding, C., Pirogov, A.: Determinization of {Büchi} {Automata}: {Unifying}
  the {Approaches} of {Safra} and {Muller}-{Schupp}. Schloss Dagstuhl -
  Leibniz-Zentrum für Informatik GmbH, Wadern/Saarbruecken, Germany (2019),
  \url{http://drops.dagstuhl.de/opus/volltexte/2019/10696/}

\bibitem{Marti2021}
Marti, J., Venema, Y.: A focus system for the alternation-free
  {\(\mu\)}-calculus. In: {TABLEAUX} 2021, Proceedings. pp. 371--388. Lecture
  Notes in Computer Science, Springer (2021).
  \doi{10.1007/978-3-030-86059-2\_22}

\bibitem{Niwinski1996}
Niwinski, D., Walukiewicz, I.: Games for the mu-{Calculus}. Theor. Comput. Sci.
   \textbf{163}(1\&2),  99--116 (1996). \doi{10.1016/0304-3975(95)00136-0}

\bibitem{Safra1988}
Safra, S.: On the complexity of $\omega$-automata. In: Proceedings of the 29th
  Symposium on the Foundations of Computer Science. pp. 319--327. IEEE Computer
  Society Press (1988)

\bibitem{Simpson2017}
Simpson, A.: Cyclic {Arithmetic} {Is} {Equivalent} to {Peano} {Arithmetic}. In:
  Foundations of {Software} {Science} and {Computation} {Structures}. pp.
  283--300. Lecture {Notes} in {Computer} {Science}, Springer (2017).
  \doi{10.1007/978-3-662-54458-7_17}

\bibitem{Sprenger03}
Sprenger, C., Dam, M.: On the {Structure} of {Inductive} {Reasoning}:
  {Circular} and {Tree}-{Shaped} {Proofs} in the μ{Calculus} (2003).
  \doi{10.1007/3-540-36576-1_27}

\bibitem{Stirling2014}
Stirling, C.: A tableau proof system with names for modal mu-calculus. In:
  Voronkov, A., Korovina, M. (eds.) HOWARD-60. A Festschrift on the Occasion of
  Howard Barringer's 60th Birthday. EPiC Series in Computing, vol.~42, pp.
  306--318. EasyChair (2014). \doi{10.29007/lwqm}

\bibitem{Studer2008}
Studer, T.: On the proof theory of the modal mu-calculus. Studia Logica
  \textbf{89}(3),  343--363 (aug 2008). \doi{10.1007/s11225-008-9133-6}

\bibitem{Walukiewicz2000}
Walukiewicz, I.: Completeness of {K}ozen's axiomatisation of the propositional
  {$\mu$}-calculus. Information and Computation  \textbf{157},  142--182 (2000)

\bibitem{Wilke2001}
Wilke, T.: Alternating tree automata, parity games, and modal $\mu$-calculus.
  Bulletin of the Belgian Mathematical Society  \textbf{8},  359--391 (2001)

\end{thebibliography}
